\documentclass[a4paper]{article}

\usepackage{a4wide}

\usepackage{amsmath}
\usepackage{amssymb}
\usepackage{amsthm}
\usepackage{hyperref}
\usepackage{etoolbox}

\usepackage{tikz}
\usepackage{subfig}

\theoremstyle{definition}
\newtheorem{theorem}{Theorem}
\newtheorem{corollary}{Corollary}
\newtheorem{lemma}{Lemma}
\newtheorem{proposition}{Proposition}

\theoremstyle{definition}
\newtheorem{definition}{Definition}
\newtheorem{example}{Example}
\newtheorem{remark}{Remark}

\newcommand{\acts}{F}
\newcommand{\alg}{\mathcal{B}}
\newcommand{\algI}{\mathcal{I}}
\newcommand{\algJ}{\mathcal{J}}
\newcommand{\bang}{\mathord{!}}

\newcommand{\catD}{\mathcal{D}}
\newcommand{\catE}{\mathcal{E}}
\newcommand{\choice}{\Gamma}
\newcommand{\choicePlayers}{\Players{\choice}}
\newcommand{\CS}{\mathsf{CS}}

\newcommand{\dist}{\mathbb{D}}
\newcommand{\embedded}[1]{\maxi({#1})}

\newcommand{\id}{\mathsf{id}}
\newcommand{\justchoice}{\mathcal{C}}
\newcommand{\justmaxi}{m}
\newcommand{\justprel}{\mathcal{P}}
\newcommand{\maxi}{\lambda}
\newcommand{\maxiPlayers}{\Players{\maxi}}

\newcommand{\nat}{\mathbb{N}}
\newcommand{\op}{\mathsf{op}}
\newcommand{\powerset}{\mathbb{P}}
\newcommand{\Players}[1]{{#1}^\heartsuit}
\newcommand{\pref}{\preccurlyeq}
\newcommand{\strpref}{\prec}
\newcommand{\prel}{\Pi}
\newcommand{\prelPlayers}{\Players{\prel}}
\newcommand{\Set}{\mathsf{Set}}
\newcommand{\U}{\mathcal{U}}
\newcommand{\Unc}{\mathsf{Unc}}
\newcommand{\supp}{\mathsf{supp}}
\newcommand{\X}{\mathcal{X}}

\renewcommand{\iff}{\qquad \mbox{iff} \qquad}

\newtoggle{preprint}
\toggletrue{preprint}

\title{Choice Structures in Games}

\author{Paolo Galeazzi and Johannes Marti}

\begin{document}
\maketitle


\begin{abstract}
\noindent Following the decision-theoretic approach to game theory, we
extend the analysis of \cite{EpsteinWang96} and \cite{DiTillio08} from
hierarchies of preference relations to hierarchies of choice functions.
We then construct the universal choice structure containing all these
choice hierarchies, and show how the universal preference structure from \cite{DiTillio08} is
embedded in it.
\end{abstract}

\section{Introduction}
\label{intro}

The present work focuses on the foundations of that part of the theory of games that deals with the players' interactive beliefs and subjective preferences. We work here in the broadest setup, where the players' beliefs and preferences are expressed in terms of choice functions with no further properties assumed. This generalizes the existing approaches in the literature, where the players' beliefs and preferences are represented by probability and utility functions or by preference relations. We show that even in our general setup it is possible to build a universal type structure - which we call universal choice structure - containing all the coherent hierarchies of choice functions. By being universal, we mean that such structure is terminal, in the sense of category theory, and complete, in the sense of \cite{Brandenburger03}. Furthermore, the universal choice structure is non-redundant (in the sense of \cite{MertensZamir}) and topology-free (in the sense of \cite{HeifetzSamet1998}). All these notions are discussed in more detail in Section \ref{discussion}.


The motivation for the present work comes from the limitations imposed
by employing order relations to represent the players' preferences,
as already recognized by \cite{Epstein97}.
Generalizations of probabilistic type structures, called preference structures, have been introduced in \cite{EpsteinWang96,DiTillio08} and \cite{GanguliHeifetz2016}, where types are defined as hierarchies of interactive preference relations rather than as hierarchies of interactive probabilistic beliefs.
However, not all important decision
criteria are representable by order relations and, consequently, preference
structures are of no help to the study of the behavioral implications
when rationality is defined in terms of such criteria. A famous example
is regret minimization, which violates the principle of independence
of irrelevant alternatives and hence transitivity too (see Section~\ref{SubExample} below). As a consequence,
regret-minimizing preferences cannot be represented by order relations
and have been axiomatized by means of choice functions (see \cite{hayashi08}
and \cite{Stoye2011}).

The present paper aims to build the interactive epistemic structures
required to take into account cases of decision criteria like regret
minimization, where preference orders are insufficient to represent
the players' rationality. Such motivation justifies the shift from
hierarchies of interactive preference orders to hierarchies of interactive
choice functions.

In addition to the conceptual motivation of this work, choice structures
turn out to be the most general models for interactive epistemology and
epistemic game theory introduced so far. As preference orders are
special cases of choice functions, the universal choice structure embeds
the universal preference structure, which in turn embeds the universal probabilistic
type structure.

On a more technical side, we prove the formal results  of this paper
using notions from the theory of coalgebras and category theory. This
allows for a modular approach, which may also shed new light on the
construction of both the universal probabilistic type structure and the universal
preference structure.

The paper is structured as follows. The following two Subsections \ref{Related literature} and \ref{SubExample} survey the related literature and provide a motivational example for this work, respectively. Section \ref{preliminaries} introduces the setting and the preliminary notions that are used in Section \ref{choice structures} to construct hierarchies of choice functions and the universal choice structure. Section \ref{preference structures} shows how the universal preference structure is embedded into the universal choice structure, and Section \ref{discussion} concludes. 

\subsection{Related literature} \label{Related literature}

As already pointed out above, \cite{EpsteinWang96}  generalize the
notion of a type from a hierarchy of probabilistic beliefs to a
hierarchy of interactive preference relations and introduce preference
structures accordingly. \cite{chen2010} then proves that any structure \`{a} la
\cite{EpsteinWang96} can be embedded into their preference structure by
a morphism which is unique. Starting from different premises than
\cite{EpsteinWang96}, \cite{DiTillio08} also constructs a preference
structure which he shows to be universal and non-redundant. In the same
spirit as \cite{EpsteinWang96} and \cite{DiTillio08}, \cite{GanguliHeifetz2016} show the existence of the universal structure for a larger class of preferences than in \cite{EpsteinWang96}. Our results are thus a further generalization from preference structures to choice structures tout court, in order to be even more liberal about the players' rationality and decision criteria.

Other related work, more for the formal techniques employed than in the spirit, is \cite{heinsalu2014}. There, the author provides the construction of a universal type structure with unawareness by means of tools from category theory. Different from our case, however, for his construction to succeed it suffices to prove that his type structures with unawareness are coalgebras for an appropriately defined functor, for which
the results of \cite{heinsalu2014} directly follow from those in \cite{Moss04,Viglizzo05}. 
In our case, instead, a more extensive use of category theory is necessary, as our choice structures are coalgebras for a new different functor, whose categorical properties have not been investigated yet.

This last observation also relates our results to the work by \cite{Moss04,Viglizzo05}, in that they are the first to formulate classic probabilistic type structures as coalgebras for the probability measure functor and to show that the universal type structure is the terminal coalgebra in the category of probabilistic type structures appropriately defined. We show here that the terminal coalgebra also exists for choice structures.

\subsection{A motivational example} \label{SubExample}

Consider the following game, that Ida is playing as the row player
and Joe as the column player.

\medskip{}

\hfill{}%
\begin{tabular}{c|cc}
 & $l$ & $r$\tabularnewline
\hline 
$u$ & 5;1 & 0;0\tabularnewline
$m$ & 3;2 & 0;1\tabularnewline
$c$ & 1;1 & 3;0\tabularnewline
$d$ & 1;2 & 2;3\tabularnewline
\end{tabular}\hfill{}

\medskip{}

\noindent Ida has two dominated actions: $m$ and $d$. Specifically, the mixed
action $pu+(1-p)c$ strongly dominates $m$ for $p\in(1/2,1)$ and
$d$ for $p\in(0,1/3)$. Therefore, no probabilistic belief on Joe's
actions can justify the choice of either $m$ or $d$ in terms of
expected utility maximization: For every probabilistic belief of Ida, either $u$ or $c$ will have higher expected utility than $m$ as well as $d$. 
This is evident from Figure~\ref{subfig:normal}, where the horizontal axis represents the probability of Joe playing $r$ and the vertical axis represents Ida's expected utility.
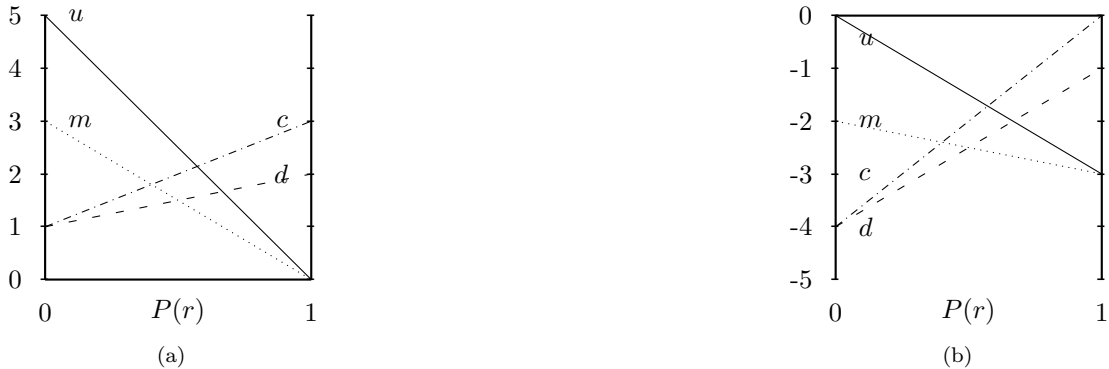
\begin{figure}
\subfloat[]{\begin{tikzpicture}[scale=.7]

\draw [-, thick]  (0,0) to (5,0);
\draw [-, thick] (0,0) to (0,5);

\draw [-, thick] (5,0) to (5,5);

\node (1y2) [] at (0,0)
  [label=left:{0}] [label=below:{0}] {-};
\node (1y2) [] at (2.5,0)
  [label=below:{$P(r)$}] {};
\node (1y2) [] at (0,1)
  [label=left:{1}] {-};
\node (1y2) [] at (0,2)
  [label=left:{2}] {-};
\node (1y2) [] at (0,3)
  [label=left:{3}] [label=right:{$m$}] {-};
\node (1y2) [] at (0,4)
  [label=left:{4}] {-};
\node (1y2) [] at (0,5)
  [label=left:{5}] [label=right:{$u$}] {-};

\node (1y0) [] at (5,0)
  [label=below:{1}] {-};
\node (1y2) [] at (5,1)
  [label=left:{}] {-};
\node (1y2) [] at (5,2)
  [label=left:{$d$}] {-};
\node (1y2) [] at (5,3)
  [label=left:{$c$}] {-};
\node (1y2) [] at (5,4)
  [label=right:{}] {-};
\node (1y2) [] at (5,5)
  [label=right:{}] {-};

\draw [-] (0,5) to (5,0);
\draw [-, dotted] (0,3) to (5,0);
\draw [-, dashdotted] (0,1) to (5,3);
\draw [-, loosely dashed] (0,1) to (5,2);

\end{tikzpicture}\label{subfig:normal}}\hfill{}\subfloat[]{\begin{tikzpicture}[scale=.7]

\draw [-, thick]  (0,5) to (5,5);
\draw [-, thick] (0,0) to (0,5);

\draw [-, thick] (5,0) to (5,5);

\node (1y2) [] at (0,0)
  [label=left:{-5}] [label=below:{0}] {-};
\node (1y2) [] at (2.5,0)
  [label=below:{$P(r)$}] {};
\node (1y2) [] at (0,1)
  [label=left:{-4}] [label=right:{$d$}] {-};
\node (1y2) [] at (0,2)
  [label=left:{-3}] [label=right:{$c$}] {-};
\node (1y2) [] at (0,3)
  [label=left:{-2}] [label=right:{$m$}] {-};
\node (1y2) [] at (0,4)
  [label=left:{-1}] [label=above right:{$u$}] {-};
\node (1y2) [] at (0,5)
  [label=left:{0}] [label=right:{}] {-};

\node (1y0) [] at (5,0)
  [label=below:{1}] {-};
\node (1y2) [] at (5,1)
  [label=left:{}] {-};
\node (1y2) [] at (5,2)
  [label=left:{}] {-};
\node (1y2) [] at (5,3)
  [label=left:{}] {-};
\node (1y2) [] at (5,4)
  [label=right:{}] {-};
\node (1y2) [] at (5,5)
  [label=right:{}] {-};

\draw [-] (0,5) to (5,2);
\draw [-, dotted] (0,3) to (5,2);
\draw [-, dashdotted] (0,1) to (5,5);
\draw [-, loosely dashed] (0,1) to (5,4);

\end{tikzpicture}\label{subfig:regret}}

\caption{On the left, the representation of Ida's expected utilities associated with the actions in the game above. On the right, the representation of Ida's expected negative regrets associated with the same actions.}
\label{fig: example}
\end{figure}
 Actions that are consistent with the players' rationality and common belief in rationality are called rationalizable. 
In the example, the
only action profile that is rationalizable with expected utility maximization is $(u,l)$, since action $r$
is no longer rational when $d$ is eliminated.

When, in the wake of decision-theoretic developments, classic probabilistic
beliefs are generalized to possibly non-probabilistic beliefs (e.g.
\cite{GilSch89,Schmeidler89}), the game-theoretic notions of rationality, dominance and equilibrium
have to be reconsidered accordingly. Examples of these advancements in game theory can be found e.g. in \cite{Klibanoff96,Lo1996,Marinacci00,KajiiUi05,SchlagRENOU2010,halpass12,BattCerrMM15, Trost2019}. 
For the sake of explanation,
let us take the case where beliefs are represented by sets of probability
distributions as in the multiple-prior (MP, henceforth) model of \cite{GilSch89}.
In this case, each action is associated with a set of expected utilities,
and a natural notion of rational choice consists in picking an action
with the highest minimal expected utility. When maxmin expected utility
is substituted for expected utility maximization as the notion of
rationality in the presence of non-probabilistic beliefs, a question naturally arises: How is dominance defined in this setting?

An answer is offered by \cite{Epstein97}, who establishes a correspondence
between iterated elimination of MP-dominated actions and MP-rationalizability.
To exemplify, consider again the game above between Ida and Joe, and
suppose that both players are expected utility maximinimizers: What
are then the behavioral implications of rationality and common belief
in rationality? Figure~\ref{subfig:normal} shows that actions $u,c$
and $d$ are justifiable for Ida: $u$ and $c$ are still best replies
to some probabilistic belief, while $d$ is now a possible best reply
to some non-probabilistic belief, e.g., the set of probability distributions
that coincides with the simplex over Joe's actions. The MP-rationalizable
action profiles are therefore the members of the Cartesian product
$\{u,c,d\}\times\{l,r\}$.

The results about MP-dominance and MP-rationalizability in \cite{Epstein97} are based
on preference structures, i.e., type structures whose elements consist
in hierarchies of interactive, reflexive and transitive preference
relations over Savage-style acts. The space of all coherent preference
hierarchies, i.e., the universal preference structure, is thus foundational
to the results about iterated dominance and rationalizability for
decision criteria that are representable by reflexive and transitive
preference relations over acts, such as maxmin expected utility and
all other noteworthy criteria in \cite{Epstein97}. As mentioned above, universal preference structures have been constructed by both \cite{EpsteinWang96} and \cite{DiTillio08}.

Once multiple decision criteria are introduced for single-agent problems, it is also natural to think of situations where different individuals adhere to different criteria in playing games. In such cases, we may have for instance Ida playing the game and choosing actions according to criterion A, while Joe is making his choices according to criterion B. Uncertainty about the opponent's criterion therefore enters the players' epistemic state and spreads to higher-order levels too: Ida may be uncertain about Joe's criterion and about Joe's uncertainty relative to her criterion, and so on. Preference structures provide a formal framework to express such interactive higher-order uncertainty about the players' decision criteria.

Consider the game above again, but now suppose that both Ida and Joe
are regret minimizers. Figure~\ref{subfig:regret} helps picture
the situation, where each action is plotted in terms of its expected negative
regret. Taking advantage of the fact that the minimization of the maximal (positive)
regret is equivalent to maxmin negative regret, Figure~\ref{subfig:regret}
shows that action $m$ is now a possible best reply (e.g., to the set
of probability distributions coinciding with the simplex over Joe's
actions), whereas action $d$ is no longer a best reply to any possible
belief of Ida. Joe would consequently not play action $r$, and
the only regret-rationalizable profile is thus $(u,l)$.

As already pointed out, however, regret minimization violates the independence
of irrelevant alternatives and is therefore not representable by a preference relation. To see it, we can make use of the game above again.
Suppose for instance that Ida is a regret minimizer and her belief
is represented by a convex compact set of probability distributions
assigning action $r$ a lower probability of 0.25 and an upper probability
of 1 (see Figure~\ref{fig: example-2}).
\begin{figure}
\subfloat[]{\begin{tikzpicture}[scale=.7]

\draw [-, thick]  (0,5) to (5,5);
\draw [-, thick] (0,0) to (0,5);

\draw [-, thick] (5,0) to (5,5);

\node (1y2) [] at (0,0)
  [label=left:{-5}]  [label=below:{0}] {-};
\node (1y2) [] at (2.5,0)
  [label=below:{$P(r)$}] {};
\node (1y2) [] at (0,1)
  [label=left:{-4}] [label=right:{$d$}] {-};
\node (1y2) [] at (0,2)
  [label=left:{-3}] [label=right:{$c$}] {-};
\node (1y2) [] at (0,3)
  [label=left:{-2}] [label=right:{$m$}] {-};
\node (1y2) [] at (0,4)
  [label=left:{-1}] [label=above right:{$u$}] {-};
\node (1y2) [] at (0,5)
  [label=left:{0}] [label=right:{}] {-};

\node (1y0) [] at (5,0)
  [label=below:{1}] {-};
\node (1y2) [] at (5,1)
  [label=left:{}] {-};
\node (1y2) [] at (5,2)
  [label=left:{}] {-};
\node (1y2) [] at (5,3)
  [label=left:{}] {-};
\node (1y2) [] at (5,4)
  [label=right:{}] {-};
\node (1y2) [] at (5,5)
  [label=right:{}] {-};

\draw [-] (0,5) to (5,2);
\draw [-, dotted] (0,3) to (5,2);
\draw [-, dashdotted] (0,1) to (5,5);
\draw [-, loosely dashed] (0,1) to (5,4);

\draw [-, loosely dotted] (1.25,0) to (1.25,5);

\node (1y2) [] at (1.25,0)
  [label=below:{.25}] {};

\end{tikzpicture}\label{subfig:normal-2}}\hfill{}\subfloat[]{\begin{tikzpicture}[scale=.7]

\draw [-, thick]  (0,5) to (5,5);
\draw [-, thick] (0,0) to (0,5);

\draw [-, thick] (5,0) to (5,5);

\node (1y2) [] at (0,0)
  [label=left:{-5}] [label=below:{0}] {-};
\node (1y2) [] at (2.5,0)
  [label=below:{$P(r)$}] {};
\node (1y2) [] at (0,1)
  [label=left:{-4}] {-};
\node (1y2) [] at (0,2)
  [label=left:{-3}] {-};
\node (1y2) [] at (0,3)
  [label=left:{-2}] [label=above right:{$c$}] [label=right:{$d$}] {-};
\node (1y2) [] at (0,4)
  [label=left:{-1}] [label=above right:{$m$}] {-};
\node (1y2) [] at (0,5)
  [label=left:{0}] [label=right:{}] {-};

\node (1y0) [] at (5,0)
  [label=below:{1}] {-};
\node (1y2) [] at (5,1)
  [label=left:{}] {-};
\node (1y2) [] at (5,2)
  [label=left:{}] {-};
\node (1y2) [] at (5,3)
  [label=left:{}] {-};
\node (1y2) [] at (5,4)
  [label=right:{}] {-};
\node (1y2) [] at (5,5)
  [label=right:{}] {-};

\draw [-, dotted] (0,5) to (5,2);
\draw [-, dashdotted] (0,3) to (5,5);
\draw [-, loosely dashed] (0,3) to (5,4);

\draw [-, loosely dotted] (1.25,0) to (1.25,5);

\node (1y2) [] at (1.25,0)
  [label=below:{.25}] {};


\end{tikzpicture}\label{subfig:regret-2}}

\caption{On the left, the representation of Ida's expected negative regrets associated with the actions in the game above. On the right, the representation of Ida's expected negative regrets associated with the actions in the game above, when action $u$ is no longer available.}
\label{fig: example-2}
\end{figure}
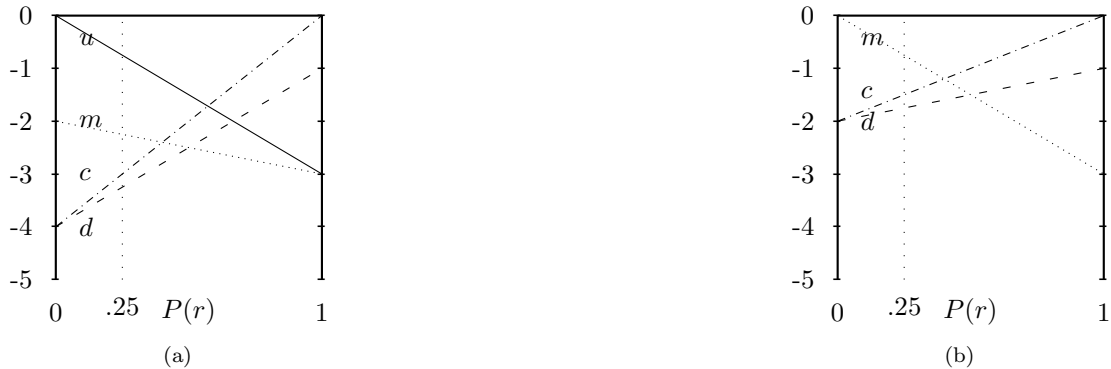
 For this specific belief, Ida is then indifferent between actions
$u,m$ and $c$, when she can choose among her four original actions
(Figure~\ref{subfig:normal-2}). However, when action $u$ is no
longer available, Ida will go for action $c$ (Figure~\ref{subfig:regret-2}):
actions $m$ and $c$ are no longer equivalent, in violation of the
independence of irrelevant alternatives. This choice pattern cannot
be encoded by a preference order, and a choice function $C$ such
that
\[
\begin{array}{ccc}
C(\{u,m,c,d\})=\{u,m,c\} &  & C(\{m,c,d\})=\end{array}\{c\}
\]
has instead to be employed. These cases never arise for maxmin expected
utility, in that the relative order among possible alternatives does
not change when actions are added or removed. This is essentially
the reason why criteria such as maxmin expected utility can be axiomatized
by preference orders (\cite{GilSch89}), while context-dependent criteria
like regret minimization require the use of choice functions (\cite{hayashi08}
and \cite{Stoye2011}). In interactive contexts, this  motivates the generalization from preference structures to choice structures.

\section{Preliminaries}
\label{preliminaries}


In this section we introduce the mathematical notions that will be  employed in game-theoretic contexts in the following sections.

\subsection{Uncertainty spaces}
\label{uncertainty spaces}

We are working with \emph{uncertainty spaces}, or simply \emph{spaces},
$(X,\alg)$, where $X$ is any set, whose elements are called
\emph{states}, and $\alg \subseteq \powerset X$ is an algebra of subsets
of $X$, where we denote by $\powerset X$ the powerset of $X$. The elements of
$\alg$ are called \emph{measurable sets} or \emph{events}.\footnote{That
$\alg$ is an algebra means that it is closed under finite intersections
and complements. Thus, an uncertainty space $(X,\alg)$ is almost a
measurable space, but without requiring closure under countable
intersections.}
We often write just $X$ for an uncertainty space $(X,\alg)$. Any set $X$
can be considered as a trivial uncertainty space $(X,\alg)$ in which
every subset is measurable, meaning that $\alg = \powerset X$ is the
\emph{discrete algebra} containing all subsets of $X$.

A morphism $\varphi$ from an uncertainty space $X$ to an uncertainty
space $Y$ is any function $\varphi : X \to Y$ that is \emph{measurable},
meaning that $\varphi^{-1}[E] = \{x \in X \mid \varphi(x) \in E\}$ is
measurable in $X$ whenever $E$ is measurable in $Y$.
For every uncertainty space $X = (X,\alg)$ we use $\id_X$ to denote the
measurable identity function on $X$, that is, $\id_X : X \to X, x \mapsto
x$.

Two uncertainty spaces $X$ and $Y$ are \emph{isomorphic}, written as $X
\simeq Y$, if there is a bijective function $\varphi : X \to Y$ such
that both $\varphi$ and its inverse $\varphi^{-1} : Y \to X$ are
measurable, and $\varphi \circ \varphi^{-1} = \id_Y$ and $\varphi^{-1} \circ \varphi
= \id_X$.

Given two uncertainty spaces $X$ and $Y$ we can define their product to
be the uncertainty space $X \times Y$ whose states are all pairs $(x,y)$
where the first component is a state in $X$ and the second component is
a state in $Y$. The measurable sets of states are generated by taking
finite unions and complements of cylinders, that is, sets of the form $U
\times Y$ and $X \times V$ for measurable $U$ in $X$ and $V$ measurable
in $Y$. It is clear that with this algebra the projections $\pi_1 : X
\times Y \to X, (x,y) \mapsto x$ and $\pi_2 : X \times Y \to Y, (x,y)
\mapsto y$ are measurable functions. Moreover given two measurable
functions $\varphi : X \to X'$ and $\psi : Y \to Y'$ we use $\varphi \times \psi : X
\times Y \to X' \times Y'$ to denote the measurable function, which is
defined such that $(\varphi \times \psi)(x,y) = (\varphi(x),\psi(y))$ for all $(x,y) \in X
\times Y$. It is not hard to check that this $\varphi \times \psi$ is indeed
measurable.

\subsection{Savage acts}

Fix a non-empty set $Z$ of \emph{outcomes}. A \emph{Savage act}, or
just \emph{act}, for an
uncertainty space $X$ is a measurable finite step function $f : X \to
Z$. That $f$ is a finite step function means that its range $f[X] =
\{f(x) \in Z \mid x \in X\}$ is finite. We also assume that the set $Z$
carries the discrete algebra in which all sets are measurable. Using
that the measurable sets in $X$ are closed under finite unions, one
easily checks that a finite step function $f : X \to Z$ is measurable precisely if $f^{-1}[\{z\}]$ is
measurable for all $z \in Z$. We write $\acts X$ for the set of all acts
for some uncertainty space $X$.
For every measurable function $\varphi : X \to Y$ we obtain a function
$\acts \varphi : \acts Y \to \acts X, f \mapsto f \circ \varphi$.

Our notion of a Savage act, where $Z$ is any set and $f : X \to Z$ is a
finite step function, is then compatible with the settings from both
\cite{DiTillio08} and \cite{Stoye2011}. In the work on preference
structures from \cite{DiTillio08}, it is assumed that $Z$ is finite and
thus all acts are automatically finite step functions. In the framework
from \cite{Stoye2011}, acts are measurable finite step functions $f : X
\to \Delta \mathcal{Z}$, where $\Delta \mathcal{Z}$ is the set of all
probability distributions with finite support over a set $\mathcal{Z}$
of outcomes. This approach can be recovered in our setting by
instantiating the set $Z$ with $\Delta \mathcal{Z}$. 

In the following we are making the assumption that $Z$ is finite, following \cite{DiTillio08}. However, our results on the existence of the universal choice structure hold also for arbitrary, possibly infinite, sets $Z$. We address this matter in more detail with Remark~\ref{r:finiteness} in \iftoggle{preprint}{Section~\ref{definition mu} of the appendix}{\ref{definition mu}}.


\subsection{Choice functions}
\label{subsec choice functions}

A \emph{choice function} over a set $X$ is a function $C$ that maps
every finite subset $F \subseteq X$ to one of its subsets $C(F)
\subseteq F$.
For any set $X$ we write $\justchoice X$ for the set of all choice
functions over $X$. Even if $X$ is just a set, without a notion of
measurable subset, we take $\justchoice X$ to be the uncertainty space
in which a set is measurable if it can be generated by taking finite
intersections and complements of sets of the form
\[
 B^K_L = \{C \in \justchoice X \mid C(K) \subseteq L\},
\]
for some finite $K,L \subseteq X$ with $L \subseteq K$.

Given any function $f : X \to Y$ we define the function $\justchoice f :
\justchoice Y \to \justchoice X$ by setting $\justchoice f (C) = C^f$,
where $C^f$ is the choice function mapping a finite $K \subseteq X$ to
\[
 C^f (K) = f^{-1}[C(f[K])] \cap K .
\]
We use $f[K] \subseteq Y$ to denote the direct image $f[K] = \{f(x) \in
Y \mid x \in K\}$ of $K$. One can show that the function $\justchoice f
: \justchoice Y \to \justchoice X$ is measurable. To see this one first
checks that for all measurable sets of the form $B^K_L$ with $K \subseteq L$
\begin{align*}
 (\justchoice f)^{-1}[B^K_L] 
 & = \{C \in \justchoice Y \mid f^{-1}[C(f[K])] \cap K \subseteq L \} \\
 & = \{C \in \justchoice Y \mid C(f[K]) \subseteq \hat{L} \} =
B^{f[K]}_{\hat{L}}, 
\end{align*}
where $\hat{L} = \{y \in f[K] \mid f^{-1}[\{y\}] \cap K \subseteq L\}$.
This then extends to arbitrary measurable sets because intersections and
complements are preserved under taking inverse images.


\subsection{Choice functions over Savage acts}
\label{choice over acts}

We then consider choice functions over Savage acts for some uncertainty space
$X$. For every uncertainty space $X$ we define the uncertainty space
$\choice X = \justchoice \acts X$ to be the uncertainty space of all
choice functions over Savage acts for $X$. Moreover, for every measurable
function $\varphi : X \to Y$ we can define the measurable function
$\choice \varphi = \justchoice \acts \varphi : \choice X \to \choice Y$.
By unfolding the definitions of $\justchoice$ and $\acts$, we can describe the choice
function $\choice \varphi (C) = C^\varphi$ for any $C \in \choice X$
more concretely: It maps a finite set of Savage acts $K \subseteq \acts
Y$ to the set
\[
 C^\varphi (K) = \{f \in K \mid f \circ \varphi \in C(\{g \circ \varphi
\mid g \in K\}) \}.
\]

\section{Choice structures}
\label{choice structures}

In this section, the notions introduced above are applied to the game-theoretic context that we are interested in. To keep things notationally simple, we focus here on interactive situations with only two
players, Ida and Joe. It is straightforward to adapt our setting to
more than two players, but this would  introduce additional notational complications.

\subsection{Choice structures}
\label{subsection choice structures}

As we are working with 
two-player games, the basic uncertainty of Ida is
just over the fixed finite set $A_j$ of Joe's actions
and, similarly, the basic uncertainty of Joe is over the fixed finite
set of Ida's actions $A_i$. We thus obtain the following definition of
a choice structure:
\begin{definition} \label{definition choice structures}
A \emph{choice structure} is a tuple $\X = (T_i,T_j,\theta_i,\theta_j)$
consisting of:
\begin{itemize}
  \item uncertainty spaces $T_i$ and $T_j$ of \emph{types} for Ida and
Joe, and
  \item measurable functions $\theta_i : T_i \to \choice (A_j
\times T_j)$ and $\theta_j : T_j \to \choice (A_i \times
T_i)$.
\end{itemize}

A \emph{morphism} $\alpha : \X \to \X'$ from a choice structure $\X =
(T_i,T_j,\theta_i,\theta_j)$ to a choice structure $\X' =
(T'_i,T'_j,\theta'_i,\theta'_j)$ consists of two measurable functions
$\alpha_i : T_i \to T'_i$ and $\alpha_j : T_j \to T'_j$ such that 
\[
 \theta'_i \circ \alpha_i = \choice (\id_{A_j} \times \alpha_j) \circ
\theta_i \text{ and } \theta'_j \circ \alpha_j = \choice (\id_{A_i}
\times \alpha_i) \circ \theta_j.
\]
\end{definition}

A \emph{state} in $\X$ is a tuple $(a_i,a_j,t_i,t_j) \in A_i \times A_j
\times T_i \times T_j$. We also say that a \emph{state of Ida} is a pair $(a_i,t_i) \in A_i \times T_i$ and a \emph{state of Joe} is a pair $(a_j,t_j) \in  A_j \times T_j$. Notice that the actions specified by Ida's choice function $\theta_i(t_i)$ at state $(a_i,a_j,t_i,t_j)$ may be completely unrelated to the action $a_i$ actually played by Ida at that state (and likewise for Joe). In the example from Section~\ref{intro}, for instance, it is possible that at a given state Ida's type prescribes to choose $u$, whereas Ida's actual choice is $d$ (see Example~\ref{example} below for more details). This is a feature that choice structures have in common with preference structures. Without delving into interpretational issues, here we just think of both choice hierarchies and preference hierarchies as expressing choice attitudes or mental states rather than actual behavior. A
state of a player therefore consists of an action that she is actually
playing and a type that represents her mental attitude.\footnote{\cite{BattigalliDeVito2021} make a similar point in the context of a different framework, where players have beliefs about their own actions.}

A \emph{Savage act of Ida} is then a map $f:A_j \times T_j \to Z$ from states of
Joe to outcomes, and similarly for Joe. We hence model the type of a
player by her choices between Savage acts, whose outcome depends on the state
of the other player, but not on her own state. This modelling assumes
that the players do not have uncertainty about their own type. This is
different from the setting of \cite{DiTillio08}, who allows players to
be uncertain about their own type. We discuss this difference in
modelling more extensively in Section~\ref{sec:introspection}.

\begin{example} \label{example}
We provide an example of a choice structure
$\mathcal{X}=(T_i,T_j,\theta_i,\theta_j)$ for the game from the
introduction. In this example Ida has a single type $T_i=\{t_i\}$, while
Joe has two possible types $T_j=\{t_{Mm},t_{EU}\}$. Here, we interpret
Ida's type $t_i$ as a regret minimizer with belief represented as in
Figure \ref{fig: example-2}, i.e., a convex compact set of probability
distributions assigning action $r$ lower probability of 0.25 and upper
probability of 1, while Joe's type $t_{Mm}$ is a maximinimizer with full
uncertainty, i.e., not excluding any probability distribution over Ida's
actions, and Joe's type $t_{EU}$ is an expected utility maximizer
assigning probability $1/2$ to Ida playing $u$ and $1/2$ to Ida playing
$d$. The states of the choice structure are given by the Cartesian set
$A_i\times A_j\times T_i\times T_j$, i.e., 
\[
\{u,m,c,d\}\times\{l,r\}\times\{t_i\}\times\{t_{Mm},t_{EU}\}.
\]
Ida's states are then members of the set $\{u,m,c,d\}\times\{t_i\}$
and Joe's states are members of the set $\{l,r\}\times\{t_{Mm},t_{EU}\}$.
The set of outcomes $Z$ is naturally given by the outcomes of the
game, 
\[
Z = \{(5,1),(3,2),(1,1),(1,2),(0,0),(0,1),(3,0),(2,3)\}.
\]
The map $\theta_i$ then associates each of Ida's types with a choice
function over Savage acts defined on Joe's states. In the running example,
such Savage acts are 
\[
\begin{array}{ccccc}
f_{u}(l,t)=(5,1) &  & f_{u}(r,t)=(0,0) &  & \text{for both }t\in T_j\\
f_{m}(l,t)=(3,2) &  & f_{m}(r,t)=(0,1) &  & \text{for both }t\in T_j\\
f_{c}(l,t)=(1,1) &  & f_{c}(r,t)=(3,0) &  & \text{for both }t\in T_j\\
f_{d}(l,t)=(1,2) &  & f_{d}(r,t)=(2,3) &  & \text{for both }t\in T_j
\end{array}
\]
Similarly, Joe's Savage acts are the following:
\[
\begin{array}{ccc}
f_{l}(u,t_i)=(5,1) &  & f_{r}(u,t_i)=(0,0)\\
f_{l}(m,t_i)=(3,2) &  & f_{r}(m,t_i)=(0,1)\\
f_{l}(c,t_i)=(1,1) &  & f_{r}(c,t_i)=(3,0)\\
f_{l}(d,t_i)=(1,2) &  & f_{r}(d,t_i)=(2,3)
\end{array}
\]
When Ida's type $t_i$ is a regret minimizer as described above,
the choice function $C_i = \theta_i(t_i)$ associated with $t_i$
maps subsets of the set of Savage acts defined above as follows:
\[
\begin{array}{ccc}
C_i(\{f_{u},f_{m},f_{c},f_{d}\})=\{f_{u},f_{m},f_{c}\} &  & C_i(\{f_{m},f_{d}\})=\{f_{d}\}\\
C_i(\{f_{u},f_{m},f_{c}\})=\{f_{u},f_{m}\} &  & C_i(\{f_{c},f_{d}\})=\{f_{c}\}\\
C_i(\{f_{u},f_{m},f_{d}\})=\{f_{u},f_{m}\} &  & C_i(\{f_{u},f_{m}\})=\{f_{u}\}\\
C_i(\{f_{u},f_{c},f_{d}\})=\{f_{u}\} &  & C_i(\{f_{u},f_{d}\})=\{f_{u}\}\\
C_i(\{f_{m},f_{c},f_{d}\})=\{f_{c}\} &  & C_i(\{f_{u},f_{c}\})=\{f_{u},f_{c}\}\\
C_i(\{f_{m},f_{c}\})=\{f_{c}\}
\end{array}
\]
where we dispense with specifying the choice function in trivial cases
such as singletons. The definition of a choice structure would also
require $C_i$ to be defined on all other subsets of $\acts (A_j \times
T_j)$. For brevity we give the definition of $C_i$ only on subsets of
$\{f_u,f_m,f_c,f_d\}$, which are the Savage acts corresponding to actions in the game. As for Joe, we have that type $t_{Mm}$ is associated
with the choice function
\[
C_{Mm}(\{f_{l},f_{r}\})=\{f_{l}\}
\]
and type $t_{EU}$ with the choice function 
\[
C_{EU}(\{f_{l},f_{r}\})=\{f_{l},f_{r}\}.
\]
Again, we omit the definition of the choice functions on sets of
Savage acts which do not correspond to actions in the game.
\end{example}

 
\subsection{Choice hierarchies and the universal choice structure}
\label{universal choice structure}

We now introduce hierarchies of choice functions that represent the
higher-order attitudes of the players. To this aim we
define uncertainty spaces representing the players' $n$-th order
attitudes by a mutual induction on $i$ and $j$. In the base case
we set $\Omega_{i,1} = \choice A_j$ and $\Omega_{j,1} = \choice A_i$,
and for the inductive step $\Omega_{i,n + 1} = \choice(A_j \times
\Omega_{j,n})$ and $\Omega_{j,n + 1} = \choice(A_i \times
\Omega_{i,n})$.
The intuition is that the players' first order attitudes are represented
by their choices between Savage acts whose outcome depends just on the actual
action played by the opponent. Players' $(n + 1)$-th order attitudes
are represented by their choices between acts, whose outcome
depends on the actual action played by the opponent and the $n$-th
order attitudes of the opponent.

Note that the players' $(n + 1)$-th order attitudes determine their $n$-th order
attitudes. At the first level this means that the agent's choices in
$\Omega_{i,1}$ between Savage acts that depend on the opponent's action are
the same as her choices between Savage acts in $\Omega_{i,2}$, 
when they are taken as additionally depending trivially on the opponent's
first-order attitudes. This can be made precise with a measurable
\emph{coherence morphism} $\xi_{i,1} = \choice \pi_1 : \Omega_{i,2} \to
\Omega_{i,1}$, where $\pi_1 : A_j \times \Omega_{j,1} \to A_j$ is the
projection onto the first component. When $o_2 \in \Omega_{i,2}$
represents Ida's second-order attitudes then $\xi_{i,1}(o_2) \in
\Omega_{i,1}$ represents her first-order attitudes.

Analogously, we define a coherence morphism for Joe, by setting
$\xi_{j,1} = \choice \pi_1 : \Omega_{j,2} \to \Omega_{j,1}$, where
$\pi_1 : A_j \times \Omega_{j,1} \to A_j$. From now on we will not
bother with writing every equation explicitly for Ida and Joe. We just
write the version for Ida and then write ``and similarly for Joe'',
thereby meaning that the equation also holds with $i$ and $j$
interchanged.

By induction we can extend the idea of coherence to the higher levels.
Choices between acts depending on the opponent's action and $n$-th
order attitudes of the opponent are determined by choices between the
same acts taken as depending on the opponent's actions and their $(n +
1)$-th order attitudes.
Hence, we define by mutual induction 
\[
\xi_{i,n + 1} = \choice(\id_{A_j} \times \xi_{j,n}) : \choice(A_j \times \Omega_{j,n + 1}) \to \choice (A_j \times \Omega_{j,n})
\]
and similarly for Joe $\xi_{j,n + 1} = \choice(\id_{A_i} \times \xi_{i,n})$.

In the limit one can then consider countable sequences $o =
(o_1,o_2,\dots, o_n, \dots)$ with $o_n \in \Omega_{i,n}$ for all $n \in
\nat$. Moreover, we require these sequences to be coherent in the sense
that $\xi_{i,n}(o_{n+1}) = o_n$ for all $n$. One such sequence
completely describes a coherent state of Ida's higher-order attitudes at
all levels. Let $\Omega_i$ be the infinite set of all such coherent
sequences. There are projections $\zeta_{i,n} : \Omega_i \to
\Omega_{i,n}$ for every level $n \in \nat$. The set $\Omega_i$ becomes
an uncertainty space when endowed with the algebra generated from all
the subsets of the form $(\zeta_{i,n})^{-1} [O_n]$ for $n \in \nat$ and
measurable $O_n \subseteq \Omega_{i,n}$. All these notions can also be
defined analogously for Joe.

In the appendix we prove the central result about this construction,
which is that there exist the following isomorphisms:
\begin{theorem} \label{iso at omega}
 $\Omega_i \simeq \choice (A_j \times \Omega_j)$ and  $\Omega_j
\simeq \choice (A_i \times \Omega_i)$.
\end{theorem}

Using the isomorphisms $\mu_i : \Omega_i \to \choice(A_j \times \Omega_j)$ and
$\mu_j : \Omega_j \to \choice(A_i \times \Omega_i)$ from Theorem~\ref{iso at omega}, we then define the universal choice structure:
\begin{definition}
 The \emph{universal choice structure} $\U =
(\Omega_i,\Omega_j,\mu_i,\mu_j)$ consists of the uncertainty spaces
$\Omega_i$ and $\Omega_j$ of all coherent sequences of choice attitudes,
together with the measurable functions $\mu_i : \Omega_i \to \choice(A_j
\times \Omega_j)$ and $\mu_j : \Omega_j \to \choice(A_i \times
\Omega_i)$ from Theorem~\ref{iso at omega}.
\end{definition}

There is a technical difference between our presentation and the
approach that is usually taken in the literature, such as for instance
\cite{MertensZamir,BranDekel1993,EpsteinWang96,DiTillio08}.
It is common to define the $(n + 1)$-th level $\Omega_{i,n + 1}$
as consisting of all pairs $(x,y) \in \Omega_{i,n} \times \choice
\Omega_{j,n}$ that are coherent in the sense that the attitudes represented
by $x$ are consistent with the attitudes represented by $y$, in a sense
similar to our coherence morphisms $\xi_{i,n}$. In the limit $\Omega_i$ one then
considers sequences $(o_1,o_2,\dots) \in \Omega_i$ such that $o_n \in \choice
\Omega_{j,n}$ for all $n \in \nat$. Our approach is
equivalent to this approach, once coherence of the whole infinite
sequences is taken into account.



\subsection{Universality of the universal choice structure}


Every type $t \in T_i$ in any choice structure $\X =
(T_i,T_j,\theta_i,\theta_j)$ generates a coherent sequence of attitudes
in $\Omega_i$.
To see it, let us define first $\upsilon_{i,1} = \choice \pi_1 \circ
\theta_i : T_i \rightarrow \Omega_{i,1}, t \mapsto \choice
\pi_1(\theta_i(t))$, where $\pi_1 : A_j \times T_j \to A_j$ is the
projection. Similarly we define $\upsilon_{j,1} : T_j \to \Omega_{j,1}$.
We can then continue by mutual induction and set
\[
 \upsilon_{i, n + 1} = \choice(\id_{A_j} \times \upsilon_{j, n})
\circ \theta_i : T_i \to \Omega_{i, n + 1}
\]
\[
 t \mapsto \Gamma(\id_{A_j} \times \upsilon_{j,n})(\theta_i(t)).
\]
Similarly we define $\upsilon_{j, n + 1} : T_j \to \Omega_{j,n + 1}$.

One can easily verify that $\upsilon_{i,n} = \xi_{i,n} \circ
\upsilon_{i,n + 1}$ for all $n$. Hence, for each $t \in T_i$ the
infinite sequence $(\upsilon_{i,1}(t),\upsilon_{i,2}(t), \dots)$ is
coherent and we obtain a measurable map $\upsilon_i : T_i \to \Omega_i$.
Similarly we also obtain a measurable map $\upsilon_j : T_j \to
\Omega_j$. In the appendix we show that $\upsilon_i$ and $\upsilon_j$
together define a unique morphism $\upsilon$ into the universal choice structure:
\begin{theorem} \label{omega terminal}
 For every choice structure $\X$ there is a unique morphism of choice
structures $\upsilon : \X \to \U$ from $\X$ to the universal choice
structure $\U$.
\end{theorem}

\subsection{Non-redundancy of the universal choice structure}
\label{nonred}

We can also show that our universal choice structure is non-redundant.
The notion of non-redundancy goes back to \cite{MertensZamir}, who prove
that their universal type space has this property. Our definition of
non-redundancy is an adaptation of Definition~3 in \cite{DiTillio08}.

For the definition of non-redundancy we consider situations
in which we want to endow the sets of types $T_i$ and $T_j$ in a choice
structure $\X = (T_i,T_j,\theta_i,\theta_j)$ with an algebra
that is possibly distinct from the algebra they have as type spaces
in $\X$. To this aim we need to clarify our notation. Recall that we
write an uncertainty space $X$ as $(X,\alg)$, where the symbol $X$ is
used for both the underlying set of points and the uncertainty space
itself, including the algebra $\alg$. Given an algebra $\alg' \subseteq
\powerset X$, which might be distinct from $\alg$, we write $(X,\alg')$
for the uncertainty space that has the set $X$ as its underlying set and
the algebra $\alg'$ as its algebra. 

\begin{definition} \label{observables}
 Let $\X = (T_i,T_j,\theta_i,\theta_j)$ be some choice structure. Define
the \emph{algebras of observable events} $\alg^{ob}_{\X,i} \subseteq
\powerset T_i$ and $\alg^{ob}_{\X,j} \subseteq \powerset T_j$ to be the
smallest algebras such that $\alg^{ob}_{\X,i}$ contains the sets
$B^{ob}_{\X,i}(K,L)$ for all finite $K,L \subseteq \acts(A_j \times
(T_j,\alg^{ob}_{\X,j}))$ and $\alg^{ob}_{\X,j}$ contains the sets
$B^{ob}_{\X,j}(K,L)$ for all finite $K,L \subseteq \acts(A_i \times
(T_i,\alg^{ob}_{\X,i}))$, where we write
\[
 B^{ob}_{\X,i}(K,L) = \{t \in T_i \mid \theta_i(t)(K) \subseteq L\}
\quad \mbox{and} \quad B^{ob}_{\X,j}(K,L) = \{t \in T_j \mid
\theta_j(t)(K) \subseteq L\}.
\]
In \iftoggle{preprint}{Appendix~\ref{app:observables well defined}}{\ref{app:observables well defined}} we explain more carefully
why these algebras of observable events are well-defined.
\end{definition}

\begin{definition}
 An algebra $\alg \subseteq \powerset X$ on the set $X$ is
\emph{separating} on $X$ if for any $x,y \in X$ with $x \neq y$ there is
a $E \in \alg$ such that $x \in E$ and $y \notin E$. The choice
structure $\X$ is \emph{non-redundant} if $\alg^{ob}_{\X,i}$ is
separating on $T_i$ and $\alg^{ob}_{\X,j}$ separating on $T_j$.
\end{definition}

In the appendix we prove that a choice structure $\X$ is non-redundant
if and only if the unique morphism $\upsilon$ from $\X$ to the universal
choice structure $\U$ is injective in both components. Proposition~2.5
in \cite{MertensZamir}, Proposition~2 in \cite{Liu09} and Proposition~2
in \cite{DiTillio08} contain analogous results for probabilistic type
spaces and preference structures.
\begin{proposition} \label{p:chara nonred}
 A choice structure $\X$ is non-redundant if and only if both components
$\upsilon_i$ and $\upsilon_j$ of the unique morphism of choice
structures $\upsilon$ from $\X$ to the universal choice structure $\U$
are injective.
\end{proposition}

Because the unique morphism from the universal choice structure $\U$ to
itself has the identity function in both components, which is
injective, it follows
immediately from Proposition~\ref{p:chara nonred} that $\U$ is non-redundant.
\begin{corollary}
 The universal choice structure $\U$ is non-redundant.
\end{corollary}

\section{Embedding preference structures}
\label{preference structures}

In this section we discuss how our hierarchies of choice functions
relate to the hierarchies of preference relations introduced in
\cite{DiTillio08}.

\subsection{Di Tillio's preference structures}
\label{di tillio}

We start by reviewing the approach by \cite{DiTillio08} in our
notation. The fundamental notion of \cite{DiTillio08} is that of a
preference relation over a set $X$. In the following a \emph{preference
relation} ${\pref}$ over $X$ is a poset, that is, a reflexive, transitive
and anti-symmetric relation, over the set $X$. We require preference
relations to be anti-symmetric. This is different from
\cite{DiTillio08}, who requires preference relations to be just preorders, that means reflexive and transitive, but not necessarily
anti-symmetric relations. We justify this apparent loss of generality in
Remark~\ref{why posets} below. In most of our arguments anti-symmetry
does not play a role and hence they also work for preorders. Our reason
to require anti-symmetry is that in the case of preorders the embedding
from preference relations into choice functions need not be injective.

Write $\justprel X$ for the set of all preference relations over the set
$X$. The set $\justprel X$ can be turned into an uncertainty space by
generating the algebra of measurable events from sets of the form
$B_{x_1 \pref x_2} = \{{\pref} \in \justprel X \mid x_1 \pref x_2\}$ for
some $x_1,x_2 \in X$.

Every function $f : X \to Y$ gives rise to the measurable function
$\justprel f : \justprel Y \to \justprel X$, where a preference relation
$\pref$ over $Y$ maps to the preference relation $\justprel f (\pref) =
{\pref}^f$ over $X$ that is defined by
\begin{equation} \label{defining justprel}
 x_1 \pref^f x_2 \iff f(x_1) \pref f(x_2).
\end{equation}

We can then redo all the constructions from Section~\ref{choice
structures} using $\justprel$ instead of $\justchoice$. Let us sketch
how this works.
One considers preference relations over acts by considering for every
uncertainty space $X$ the space $\prel X = \justprel \acts X$. For every
measurable function $\varphi : X \to Y$ we obtain the measurable
function $\prel \varphi = \justprel \acts \varphi : \prel X \to \prel Y$
which maps a preference relation $\pref$ over $\acts X$ to the
preference relation $\pref^\varphi$ over $\acts Y$ that is defined such
that $f \pref^\varphi g$ if and only if $f \circ \varphi \pref g \circ \varphi$.

Lastly, define a \emph{preference structure} to be a tuple $\X =
(T_i,T_j,\theta_i,\theta_j)$ where $\theta_i : T_i \to \prel (A_j \times
T_j)$ and $\theta_j : T_j \to \prel (A_i \times T_i)$ are measurable
functions. A \emph{morphism} $\alpha : \X \to \X'$ from a preference
structure $\X = (T_i,T_j,\theta_i,\theta_j)$ to a preference structure
$\X' = (T'_i,T'_j,\theta'_i,\theta'_j)$ consists of two measurable
functions $\alpha_i : T_i \to T'_i$ and $\alpha_j : T_j \to T'_j$ such that
\[
 \theta'_i \circ \alpha_i = \prel (\id_{A_j} \times \alpha_j) \circ
\theta_i \text{ and } \theta'_j \circ \alpha_j = \prel (\id_{A_i} \times
\alpha_i) \circ \theta_j.
\]
Note that this definition of a preference structure differs from the one given in \cite{DiTillio08} with respect to the assumption of introspection mentioned in Section~\ref{subsection choice structures}.

The universal preference structure $\U' =
(\Omega'_i,\Omega'_j,\mu'_i,\mu'_j)$ can be defined from the limits
$\Omega'_i$ and $\Omega'_j$ of sequences
$(\Omega'_{i,n},\xi'_{i,n})_{n \in \nat}$ and
$(\Omega'_{j,n},\xi'_{j,n})_{n \in \nat}$ analogous to those that are used to approximate
the universal choice structure in the previous section. Hence, set
$\Omega'_{i,1} = \prel A_j$, $\Omega'_{j,1} = \prel A_j$ and then
inductively $\Omega'_{i,n + 1} = \prel(A_j \times \Omega'_{j,n})$ and
$\Omega'_{j,n + 1} = \prel(A_i \times \Omega'_{i,n})$. The coherence
morphism are such that $\xi'_{i,1} = \prel \pi_1$, $\xi'_{j,1} = \prel
\pi_2$ and inductively $\xi'_{i,n + 1} = \prel (\id_{A_j} \times
\xi'_{j,n})$ and $\xi'_{j,n + 1} = \prel (\id_{A_i} \times \xi'_{i,n})$.
The existence of suitable $\mu'_i$ and $\mu'_j$ and universality
properties of $\U'$ then follow from a construction that is analogous to
the one given in \iftoggle{preprint}{Appendix~\ref{proofs for choices}}{\ref{proofs for choices}}, using $\prel$ in
place of $\choice$. The properties of $\prel$ that are required for this
construction to succeed are stated in Section~3 of
\cite{DiTillio08}.

\subsection{Maximization}

We use maximization to map preference orders to choice functions. Given
a finite set of alternatives, a player with a given preference relation
chooses the most preferred alternatives of the set. Formally, this means
that given a preference relation ${\pref} \in \justprel X$ over a set
$X$ we map it to the choice function $\justmaxi_X ({\pref}) \in
\justchoice X$ that assigns to a finite set $K \subseteq X$ the set
$\justmaxi_X ({\pref})(K)$ of its maximal elements, which is defined as
follows:
\[
 \justmaxi_X({\pref})(K) = \{m \in K \mid \mbox{there is no } k \in K \mbox{ with } m \strpref k\},
\]
where $x \strpref y$ is defined to hold if and only if $x \pref y$ and not $y \pref x$.
This definition yields a measurable function $\justmaxi_X : \justprel X
\to \justchoice X$ for every set $X$. To prove that it is measurable it
suffices to show that for every basic event $B^K_L = \{C \in \justchoice
X \mid C(K) \subseteq L\}$ of $\justchoice X$ its preimage
$\justmaxi_X^{-1}[B^K_L]$ is measurable in $\justprel X$. To see that
this is the case first observe that the set of maximal elements of a set $K$
in a preference relation $\pref$ is a subset of $L$ if and only if for every
element $k \in K \setminus L$ there is some $l \in L$ such that $k \pref
l$ and not $l \pref k$. Thus we can write
\[
 \justmaxi_X^{-1}[B^K_L] = \bigcap_{k \in K \setminus L} \bigcup_{l \in
L} \left(B_{k \pref l} \setminus B_{l \pref k} \right) .
\]
Since $K$, $L$ and hence also $K \setminus L$ are finite the right side of
the above equality is a finite intersection of finite unions of
intersections of basic events with complements of basic events.

\begin{remark} \label{why posets}
 With our definition of the choice function $\justmaxi_X({\pref})$ from the preference relation $\pref$ we do not lose any generality by restricting to posets instead of preorders. For every preorder ${\pref}$ we can define the poset $\pref'$ by setting $x \pref' x'$ iff $x = x'$ or, $x \pref x'$ and not $x' \pref x$. One can show that $\pref'$ is anti-symmetric and that the maximal elements of any finite set $K \subseteq X$ in $\pref'$ are the same as the maximal elements of $K$ in $\pref$, meaning that $\justmaxi_X({\pref'})(K) = \justmaxi_X({\pref})(K)$. It follows from this that any choice function that arises from maximization in some preorder also arises from maximization in some poset. Hence, the restriction to posets does not lose any generality in the class of choice behaviors that we can account for.

There are examples in which the preorders $\pref$ and $\pref'$ are
distinct. For instance, we might take $X = \{x_1,x_2\}$ to be a two element
set and $\pref$ the total relation, in which the two elements of $X$ are
equally preferred, meaning that $x_1 \pref x_2$ and $x_2 \pref x_1$. Applying
the definition of $\pref'$ from above shows that $\pref'$ is then the
poset in which $x_1$ and $x_2$ are incomparable, meaning that neither $x_1
\pref' x_2$ nor $x_2 \pref' x_1$. We have then that $\justmaxi_X({\pref})(K)
= K = \justmaxi_X({\pref'})(K)$ for every $K \subseteq X$. Hence, there
exist two preorders that account for the same choice behavior. We show
in the following proposition that this redundancy disappears if one
restricts to posets.
\end{remark}

The next proposition shows that any two distinct preference relations
give rise via maximization to distinct choice functions. We prove this
proposition in the appendix. It relies on our assumption that preference
relations are anti-symmetric.
\begin{proposition} \label{justmaxi injective}
 The measurable function $\justmaxi_X : \justprel X \to \justchoice X$
is injective for all sets $X$.
\end{proposition}

We can extend maximization to preference relations and choice functions
over acts by defining for every space $X$ the measurable function $\maxi_X =
\justmaxi_{\acts X} : \prel X \to \choice X$. A crucial property of
$\maxi_X$ is stated in the following proposition, which we prove in the
appendix:
\begin{proposition} \label{maxi natural}
 The mapping $\maxi_X$ is natural in $X$. This means that for every
measurable function $\varphi : X \to Y$ we have that $\maxi_Y \circ \prel \varphi =
\choice \varphi \circ \maxi_X$.
\end{proposition}

\subsection{Embedding preference structures into choice structures}
\label{embedding}

We can use the mappings $\maxi_X$ to turn preference structures into choice
structures. For every preference structure $\X =
(T_i,T_j,\theta_i,\theta_j)$ define the choice structure $\embedded{\X}
= (T_i,T_j,\maxi_{A_j \times T_j} \circ \theta_i,\maxi_{A_i \times T_i}
\circ \theta_j)$. It is an easy consequence of Proposition~\ref{maxi
natural} that this embedding preserves morphisms in the sense that
whenever $\chi : \X \to \X'$ is a morphism of preference structures then
the same pair of measurable functions is also a morphism of choice
structures $\chi : \embedded{\X} \to \embedded{\X'}$.

It is also possible to characterize the class of choice structures
$\embedded{\X} = (T_i,T_j,\maxi_{A_j \times T_j} \circ
\theta_i,\maxi_{A_i \times T_i} \circ \theta_j)$ arising from some
preference structure $\X = (T_i,T_j,\theta_i,\theta_j)$. To this end,
one can use any of the axiomatizations of the class $\mathbb{P}$ of
choice functions that arise from the maximization in a poset. Such an
axiomatization is given for instance in Theorem~2.9 of
\cite{AleskerovetAl07}. We then have that a choice structure $\X' =
(T'_i,T'_j,\theta'_i,\theta'_j)$ is equal to $\embedded{\X}$ for some
preference structure $\X$ if and only if for all types $t_i \in T'_i$ and $t_j \in
T'_j$ the choice functions $\theta'_i(t_i) \in \choice (A_j \times T'_j)$
and $\theta'_j(t_j) \in \choice (A_i \times T'_i)$ are in $\mathbb{P}$.

Because $\maxi_X$ is injective for every space $X$, it follows that
whenever $\embedded{\X} = \embedded{\X'}$ for preference structures $\X$
and $\X'$ then $\X = \X'$. Therefore, any difference between preference
structures is preserved in choice structures. One can also define an
injective embedding of the universal preference structure into the
universal choice structure. Consider the unique morphism of choice
structures $\upsilon$ from $\maxi(\U')$ to $\U$ that exists by
Theorem~\ref{omega terminal}. This morphism consists of two measurable
functions $\upsilon_i : \Omega'_i \to \Omega_i$ and $\upsilon_j :
\Omega'_j \to \Omega_j$ for which we argue in the appendix that it is
injective if $A_i$ and $A_j$ are finite. Hence, we obtain the following
theorem:
\begin{theorem} \label{upsilon injective}
 The uncertainty spaces
$\Omega'_i$ and $\Omega'_j$ of all preference hierarchies are isomorphic
to two subspaces of the spaces of all choice hierarchies $\Omega_i$ and
$\Omega_j$ respectively.
\end{theorem}



\section{Discussion}
\label{discussion}

\subsection{Universality of the universal choice structure}

The universal choice structure that is defined in this paper has all of the properties that have been associated with the term ``universality'' in the literature. It is terminal in the categorical sense, meaning that for every choice structure there is a unique morphism to the universal one. It is complete in the sense that the measurable functions $\mu_i$ and $\mu_j$ of the universal choice structure are onto. The universal choice structure is also non-redundant, as we discussed above.
In particular, because we are working in a category of coalgebras, completeness and non-redundancy follow directly from terminality.

\subsection{Topological assumptions and coherence}

Our construction of the universal choice structure is based on
hierarchies of choice functions, where the levels in the hierarchies are
linked by what we call coherence morphisms. We use this terminology in
analogy to the construction of coherent probabilistic hierarchies by the
use of marginalization. However, our result is topology-free, in that we
do not need to impose topological assumptions on the spaces. This is
different from the classic
probabilistic setting, where without additional topological assumptions
the universal structure can not be built from coherent hierarchies of
beliefs \cite{HeifetzSamet1998,heifetzSamet1999,Viglizzo05}. 
Different from \cite{HeifetzSamet1998} and \cite{heifetzSamet1999}, the construction of the universal structure directly from coherent choice hierarchies in a topology-free setup succeeds in our case because, like in \cite{DiTillio08}, all the sets we deal with are endowed with an algebra, not a $\sigma$-algebra, of subsets. We therefore do not incur in the same problem as in \cite{heifetzSamet1999}, because the space of all the choice hierarchies is endowed with an algebra, not a $\sigma$-algebra, of subsets and the extension of the coherent choice hierarchies to such a space is thus unproblematic.

\subsection{Introspection}
\label{sec:introspection}

An important difference between the definitions of type structure in the
literature is how they account for a player's knowledge about her own
type. Similar to \cite{EpsteinWang96}, we follow an approach in the
spirit of \cite{BranDekel1993}, where a player has no uncertainty about her own
type. This treatment is different from the one in \cite{DiTillio08}, who
allows for players to be uncertain about their own type. We could use an
analogous approach by considering choice structures based on just one
uncertainty space $X$, whose elements should be thought of as
the equivalent of the states from Section~\ref{subsection choice structures}, and then
define a choice structure with measurable functions $\sigma :
X \to A_i \times A_j$ and $\vartheta : X \to \choice X \times \choice
X$. The first component of $\vartheta(x) =
(\vartheta_i(x),\vartheta_j(x))$ specifies the type of Ida, and the
second the type of Joe.
As for our proofs, this would mean that we would consider coalgebras for
the endofunctor in $\Unc$ that maps an uncertainty space $X$ to the
product $(A_i \times A_j) \times (\choice X \times \choice X)$, and with
the obvious behavior on arrows.

To avoid confusion it should be mentioned that the
approach of \cite{DiTillio08} is different from the one by
\cite{MertensZamir}, even though both settings use analogous mathematical structures. The difference is that \cite{MertensZamir}
impose an introspection condition that requires the players to be certain about 
their own type, whereas \cite{DiTillio08} does not.
As discussed in \cite[sec.~7.4]{Viglizzo05}, it is unclear how this condition can be formulated in coalgebraic terms for the structures from
\cite{MertensZamir} and \cite{DiTillio08}.

%
%




\subsection{The coalgebraic approach}

For the technical proofs in the appendix we make extensive use of concepts
from coalgebra, which in turn uses notions from category theory such as
the base category, the coalgebraic  functor, terminal objects and
natural transformations between  functors. The coalgebraic approach provides us with a modular language in which to describe the
construction of the universal choice structure. This comes with the
following advantages: First, it is easier to adapt the theory to a
new representation of uncertainty. Because the general construction of
the terminal coalgebra is parametric in the functor, it suffices to define the appropriate functor and check that it preserves limits of
suitable diagrams. 
Second, category theory provides an encompassing language that can be used to compare properties between different kinds of type structures, such as the concepts of completeness, non-redundancy and  the like. 
Third, the coalgebraic approach allows to abstract away from
notational complications such as for instance the indexing that is involved in
representing games with multiple players. When formulating the terminal-sequence construction, the indexing is conveniently taken care of in the base category and the coalgebraic functor.

\subsection{Interpretation of the framework}

A point deserving further investigation is the interpretation of both choice hierarchies and preference hierarchies. As a strict behavioral interpretation of such hierarchies seems to conflict with the specification of the players' actual behavior at a state, for now we viewed hierarchies as choice attitudes rather than choice behavior. Deeper analysis on this point is left to more philosophical research in the future, while here we have been primarily concerned with building the necessary technical apparatus introducing hierarchies of choice functions in situations of interactive epistemology. 

\bibliographystyle{alpha}
\bibliography{biblio,bibliocopy,choicePrince,bibcoalgebra}

\appendix



\section{Overview of the appendix}

This appendix contains the technical proofs about the constructions of
the universal choice structure from Section~\ref{choice structures} and
the embedding of the universal preference structure into the universal
choice structure from Section~\ref{preference structures}. We view these
constructions as instances of more general constructions from the theory
of coalgebras. Thereby, we follow the approach laid out in
\cite{Moss04,Viglizzo05}, who show how classical probabilistic type spaces can
be seen as coalgebras. The theory of coalgebras itself makes use of
notions from category theory. We review all relevant concepts from
category theory in \iftoggle{preprint}{Section~\ref{app:preliminaries} of this appendix}{\ref{app:preliminaries}}.
For the interested reader, more in-depth expositions of category theory can be found in
\cite{MacLane} and \cite{awod06}. For a treatment of the theory of
coalgebras we refer to \cite{Jacobs16} or the classic introduction
\cite{Rutten00}.

The cornerstones for the results in this appendix are as follows:

\textit{Choice structures as coalgebras:} The first insight is that
choice structures are coalgebras and morphisms of choice structures are
coalgebra morphisms. The definition of a coalgebra is parametric in a
fixed endofunctor $G$ on a category $\catD$. The category $\catD$ is called the ``base category'' and the functor $G$ is often called the ``type functor'', but we call it the ``coalgebraic'' functor in order to avoid confusion. We explain the
definition of a category and the notion of a functor in \iftoggle{preprint}{Sections
\ref{app:category}~and~\ref{app:functors}}{\ref{app:category}~and~\ref{app:functors}}. Most relevant for us are the
category $\Unc$ of uncertainty spaces and measurable functions and the
functor $\choice$ on $\Unc$, which we have already implicitly defined in
Section~\ref{choice over acts}. However, for technical reasons having to
do with the fact that choice structures represent the uncertainty of two
players, instead of just one, we have to use $\Unc^2$ as our base
category, which essentially captures situations involving two
uncertainty spaces at once, instead of just one as in $\Unc$. One also
has to work with the functor $\choicePlayers$, which is an adaptation of
$\choice$ to $\Unc^2$ that captures the dependence of a player's
uncertainty on the uncertainty of the other player.

\textit{Existence of the terminal coalgebra:} The next observation is
that the universal choice structure corresponds to what is called
the terminal coalgebra in the theory of coalgebras. Our aim is then to
use a general coalgebraic construction to
show the existence of the terminal coalgebra in the category of choice
structures. This construction, called the terminal sequence construction
and discussed in \cite{Worrell99}, can be seen as a coalgebraic
counterpart of the construction of the universal type space from the
space of coherent hierarchies of beliefs. It can be performed in all
categories of coalgebras for a given functor and under certain rather mild conditions on the
coalgebraic functor it is guaranteed to yield the terminal coalgebra. There are
different versions of these sufficient conditions, but mostly they
require that the coalgebraic functor preserves the limits of certain diagrams
in the base category. We review the required notion of limits in
\iftoggle{preprint}{Section~\ref{app:limits}}{\ref{app:limits}}. In the case of the functor $\choicePlayers$,
whose coalgebras are choice structures, we can only show the
preservation of a limit for a particular kind of diagram, which we call
epic countable cochains. In \iftoggle{preprint}{Section~\ref{proofs for choices}}{\ref{proofs for choices}} we review
the terminal sequence construction to check that the preservation of
limits for epic countable cochains suffices for it to yield the terminal
coalgebra.

\textit{Preservation of limits of countable epic cochains:} To guarantee
the existence of the terminal coalgebra we then show in
\iftoggle{preprint}{Section~\ref{choice preserves limits}}{\ref{choice preserves limits}} that the functor $\choicePlayers$
preserves limits of epic countable cochains. It is relatively easy to
see that this reduces to showing that $\choice$ preserves these limits.
The functor $\choice$ can be seen as the composition of two simpler
functors, the functor $F$, which maps an uncertainty space $X$ to the
set $F X$ of all Savage acts over $X$, and the functor $\justchoice$,
which maps a set $X$ to the uncertainty space $\justchoice X$ of all
choice functions over $X$. We then show that both $F$ and $\justchoice$
preserve suitable limits such that their composition $\choice$ preserves
limits of countable epic cochains. For $F$ this is a relatively easy
observation that is already implicitly part of the proofs in
\cite{DiTillio08}. That $\justchoice$ preserves limits of the right
kind is Theorem~\ref{preservation of limits}, which is the central
technical contribution of this paper.

\textit{Embedding of preference structures:} We also use the coalgebraic
approach to prove the results from Section~\ref{preference structures}
about the embedding of the universal preference structure of
\cite{DiTillio08} into the universal choice structure. The setting of
preference structures can also be reformulated in the coalgebraic
terminology simply by replacing the functor $\justchoice$ mentioned
above by a functor $\justprel$ that captures preference relations.
The embedding of preference structures into choice structures then
follows from the observation that there is a natural transformation from
$\justchoice$ to $\justprel$ that is injective at all components. We
discuss the notion of a natural transformation in \iftoggle{preprint}{Section~\ref{app:nat
trans}}{\ref{app:nat trans}}. In \iftoggle{preprint}{Section~\ref{proofs for prel}}{\ref{proofs for prel}} we explain how the existence of
such a natural transformation from $\justprel$ to $\justchoice$ gives
rise to an embedding of the universal preference structure into the
universal choice structure.

\section{Preliminaries}
\label{app:preliminaries}

We start by reviewing the basics from category theory and the theory of
coalgebras.

\subsection{Categories}
\label{app:category}

A \emph{category} $\catD$ is a collection of objects, possibly a proper
class, together with a collection of morphisms, written as $f : X \to
Y$, between any two objects $X$ and $Y$ such that for every object $X$
in $\catD$ there is an \emph{identity} morphism $\id_X : X \to X$ and
for all objects $X$, $Y$ and $Z$, and morphisms $f : X \to Y$ from $X$
to $Y$ and $g : Y \to Z$ from $Y$ to $Z$ their \emph{composition} $g
\circ f : X \to Z$ is a morphism in $\catD$ from $X$ to $Z$. The
identities and compositions are required to satisfy the identity law
that $f \circ \id_X = f = \id_Y \circ f$ for all morphism $f : X \to Y$
and the associativity law that $(h \circ g) \circ f = h \circ (g \circ
f)$ for all $f : X \to Y$, $g : Y \to Z$ and $h : Z \to U$. For a
morphism $f : X \to Y$ in $\catD$ we call the object $X$ the
\emph{domain} of $f$ and the object $Y$ the \emph{codomain} of $f$.

Examples of categories, which we use in this paper, are the category of
sets $\Set$, which has all the sets as its objects and functions between
them as morphisms, and the category of uncertainty spaces $\Unc$, which
has uncertainty spaces as objects and measurable functions as morphisms.
These clearly satisfy the identity and composition laws if one uses
the standard identity functions as identity morphisms and the standard
compositions of functions for the composition. A different kind of
example is the category $\CS$ of all choice structures as defined in
Section~\ref{choice structures} with morphisms of choice structures
between them. The results of this paper mostly concern the structure of
this category $\CS$.

A more technical example of a category is $\Unc^2$, which has pairs $X =
(X_0,X_1)$ of uncertainty spaces $X_0$ and $X_1$ as objects and in which
the morphisms from $X = (X_0,X_1)$ to $Y = (Y_0,Y_1)$ are pairs of
measurable functions $\varphi = (\varphi_0,\varphi_1)$. One can check
that this is a category, when one gives the obvious component-wise
definition of identities and composition. We use the category $\Unc^2$
as a convenient abstraction for the technical results of this appendix
because it concisely represents situations that concern the uncertainty
of two players. In an object $X = (X_0,X_1)$ of $\Unc^2$, the uncertainty
space $X_0$ represents the uncertainty of Ida, whereas $X_1$ represents
the uncertainty of Joe.

\subsection{Functors}
\label{app:functors}

A \emph{functor} $G$ from a category $\catD$ to a category $\catE$
is an assignment of an object $G X$ in $\catE$ to every object $X$ of
$\catD$ and an assignment of a morphism $G f : G X \to G Y$ in $\catE$
to every morphism $f : X \to Y$ in $\catD$ that preserves identities and
compositions, meaning that $G \id_X = \id_{G X}$ for all $X$ and $G(g
\circ f) = G g \circ G f$ for all $g$ and $f$. If $G$ is a functor from
a category $\catD$ to the same category $\catD$ one also calls $G$ an
\emph{endofunctor}.

A simple example of a functor is the powerset functor $\powerset$ from 
$\Set$ to $\Set$. It maps a set $X$ to its powerset $\powerset X$
and a function $f : X \to Y$ to the direct image map $\powerset f :
\powerset X \to \powerset Y$, which is defined such that $\powerset f
(U) = f[U] = \{f(x) \in Y \mid x \in U\}$ for all $U \subseteq X$. The
functor $\powerset$ is an endofunctor on the category of sets.

Another example of an endofunctor in $\Set$ is the finite distribution
functor $\dist$. It maps a set $X$ to the set $\dist X$ of all
probability distributions over $X$ that have finite support. A function
$f : X \to Y$ is mapped to the function $\dist f : \dist X \to \dist Y$
which is such that $\dist f (\mu) = \mu^f$ is the probability
distribution $\mu^f$ on $Y$ such that $\mu^f(V)
= \sum_x \mu(x)$ for $x \in f^{-1}[V]$.

A \emph{contravariant functor} from $\catD$ to $\catE$ is a functor from
$\catD$ to $\catE^\op$, where $\catE^\op$ is just like $\catE$, but the
direction of all morphisms is turned around. When $G$ is a contravariant
functor from $\catD$ to $\catE$, one usually just thinks of it as
turning morphisms around, in the sense that $G f : G Y \to G X$ for all
$f : X \to Y$. If one wants to emphasize that a functor $G$ from $\catD$
to $\catE$ is not contravariant, one calls $G$ \emph{covariant}.

An example of a contravariant functor from this paper is the mapping $F$ sending an
uncertainty space $X$ to the set $F X$ of Savage acts for $X$ and
a measurable function $\varphi : X \to Y$ to the function $F f : F Y \to
F X$. This is a
contravariant functor from the category of uncertainty spaces to the
category of sets.

Another example of a contravariant functor is the mapping $\justchoice$
from Section~\ref{subsec choice functions}, which sends a set $X$ to the
uncertainty space $\justchoice X$ of choice functions over that set and
a function $f : X \to Y$ to the measurable function $\justchoice f :
\justchoice Y \to \justchoice X$. This defines a contravariant functor
from the category of sets to the category of uncertainty spaces. It is
straightforward to check that it preserves identities. To see that
$\justchoice$ preserves the composition of morphisms, one sees from the
definitions that this amounts to checking the following equality:
\[
 f^{-1}[g^{-1}[C(g[f[K]])]] \cap K = f^{-1}[g^{-1}[C(f[g[K]])] \cap
g[K]] \cap K,
\]
for all finite sets $K \subseteq Y$, and functions $f : X \to Y$ and $g
: Y \to W$.

The mapping $\choice$ from Section~\ref{choice over acts} is a covariant
endofunctor on the category of uncertainty spaces that results from
first applying the contravariant functor $\acts$ from uncertainty spaces
to sets and then the contravariant functor $\justchoice$ to get back to
uncertainty spaces. Formally, this means that $\choice$ maps an
uncertainty space $X$ to the space $\choice X = \justchoice \acts X$ and
a measurable function $\varphi : X \to Y$ to the measurable function
$\choice \varphi = \justchoice F \varphi : \choice X \to \choice Y$.
Because the effects of the two contravariant functors $\acts$ and
$\justchoice$ cancel out, this means that $\choice$ is a covariant
functor.

The mappings $\justprel$ and $\prel$ from \cite{DiTillio08}, which we
survey in Section~\ref{di tillio}, can also be seen as functors
analogous to $\justchoice$ and $\choice$. That is, $\justprel$ is a
contravariant functor from $\Set$ to $\Unc$ and $\prel$ is the
covariant endofunctor on $\Unc$ that results from first applying the
contravariant functor $\acts$ from $\Unc$ to $\Set$ and then the
functor $\justprel$.

\subsection{Coalgebras}

We now explain how choice structures can be seen as an instance of the
more general notion of a coalgebra for an endofunctor in some category. 

To define the notion of a coalgebra we need to fix a category $\catD$
and an endofunctor $G$ on $\catD$, that is, a functor $G$ from $\catD$
to $\catD$. A \emph{$G$-coalgebra} $(X,\xi)$ is an object $X$ of $\catD$
together with a morphism $\xi : X \to G X$ in $\catD$. The class of all
$G$-coalgebras can be turned into a category by taking as morphisms from
a $G$-coalgebra $(X,\xi)$ to a $G$-coalgebra $(Y,\upsilon)$ all
morphisms $f : X \to Y$ in the category $\catD$ that satisfy the
property that $\upsilon \circ f = G f \circ \xi$.

A relatively simple example of a category of coalgebras for
the powerset functor $\powerset$ in the category $\Set$. One can see
that $\powerset$-coalgebras are essentially the same as Kripke frames
\cite{Fagin95,Blackburn02}. For any Kripke frame $(W,R)$, with a set $W$
and relation $R \subseteq W^2$, we can define a $\powerset$-coalgebra
$(W,\alpha)$ where $\alpha : W \to \powerset W$ is such that $\alpha(w)
= R[\{w\}] = \{v \in W \mid (w,v) \in R\}$ for all $w \in W$.
Conversely, every $\powerset$-coalgebra $(W,\alpha)$ also defines a
Kripke frame $(W,R)$ such where $R = \{(w,v) \mid v \in \alpha(w)\}$.
One can also verify that the coalgebra morphisms between
$\powerset$-coalgebras correspond precisely to the morphisms
between Kripke frames.

Choice structures for fixed sets $A_i$ and $A_j$ of actions can be
defined as coalgebras for a particular functor $\choicePlayers$ over the
category $\Unc^2$. The functor $\choicePlayers$ maps an object
$(X_0,X_1)$ to the object $(\choice(A_j \times X_1),\choice(A_i \times
X_0))$. The change of indices in the components is intentional, as it
corresponds to encoding attitudes about the opponent's attitudes. For
morphisms, the functor $\choicePlayers$ sends a pair of measurable
functions $(\varphi_0,\varphi_1)$ from $(X_0,X_1)$ to $(Y_0,Y_1)$ to the
pair of measurable functions $(\choice(\id_{A_j} \times \varphi_1),
\choice(\id_{A_i} \times \varphi_0))$ from $(\choice(A_j \times
X_1),\choice(A_i \times X_0))$ to $(\choice(A_j \times Y_1),\choice(A_i
\times Y_0))$.

We can view every choice structure $\X = (T_i,T_j,\theta_i,\theta_j)$,
defined according to Definition~\ref{definition choice structures}, as a
$\choicePlayers$-coalgebra $((T_i,T_j),(\theta_i,\theta_j))$ on the
object $(T_i,T_j)$ of $\Unc^2$ with the morphism $(\theta_i,\theta_j) :
(T_i,T_j) \to \choicePlayers (T_i,T_j)$. Conversely, every
$\choicePlayers$-coalgebra $((X_0,X_1),(\xi_0,\xi_1))$ contains
measurable functions $\xi_0 : X_0 \to \choice(A_j \times X_1)$ and
$\xi_1 : X_1 \to \choice(A_i \times X_0)$ and thus gives rise to a
choice structure $(X_0,X_1,\xi_0,\xi_1)$. One can also check that the
morphisms of choice structures, as defined in Definition~\ref{definition
choice structures}, correspond precisely to the coalgebra morphisms
between $\choicePlayers$-coalgebras.

Analogous to choice structures, one can view preference structures as
coalgebras for a functor. To this aim one considers the functor
$\prelPlayers$ from $\Unc^2$ to $\Unc^2$ that is defined by replacing
all occurrences of $\choice$ in the definition of $\choicePlayers$ with
$\prel$. This means that $\prelPlayers$ sends an object $(X_0,X_1)$ to
$(\choice(A_j \times X_1),\choice(A_i \times X_0))$ and a morphism
$(\varphi_0,\varphi_1)$ to $(\prel(\id_{A_j} \times \varphi_1),
\prel(\id_{A_i} \times \varphi_0))$.



\subsection{Isomorphisms and epic or monic morphisms}

A morphism $f : X \to Y$ in a category $\catD$ is an \emph{isomorphism}
if there exists a morphism $g : Y \to X$ such that $g \circ f = \id_X$
and $f \circ g = \id_Y$. One writes $X \simeq Y$ in case there exists an
isomorphism $f : X \to Y$. One can check that in the category of sets
the isomorphisms are precisely the bijective functions and in the
category $\Unc$ they are the bijective measurable functions whose
inverse is also measurable, analogously to the homeomorphism between
topological spaces. The isomorphisms in $\Unc^2$ are those pairs
$(\varphi_0,\varphi_1)$ of morphisms such that both $\varphi_0$ and
$\varphi_1$ are isomorphisms in $\Unc$.

A morphism $f : X \to Y$ is \emph{epic} if for all further morphisms
$g,h : Y \to T$ it holds that if $g \circ f = h \circ f$ then already $g
= h$. A morphism $f : X \to Y$ is \emph{monic} if for all further
morphisms $g,h : T \to X$ it holds that if $f \circ g = f \circ h$ then
already $g = h$. One can check that in the category of sets epic
morphisms are precisely the surjective functions and monic morphisms
precisely the injective functions. Similarly, in the category of
uncertainty spaces the epic and monic morphisms are the surjective
and injective measurable functions, respectively. Epic and monic morphisms in $\Unc^2$
are epic and monic in both components. That is, $(\varphi_0,\varphi_1)$
is epic if both $\varphi_0$ and $\varphi_1$ are epic, and analogously
for monic.

A functor $F$ from $\catD$ to $\catE$ \emph{preserves} an epic morphism
$f$ of $\catD$ if $F f$ is epic in $\catE$, and $F$ \emph{preserves} a monic
morphism $f$ if $F f$ is monic. Note that, because an epic morphism in
$\catE^\op$ is monic in $\catE$, a contravariant functor $F$ from
$\catD$ to $\catE$ preserves an epic morphism $f$ if $F f$ is monic in
$\catE$, and analogously for the preservation of monic morphisms.

We show in \iftoggle{preprint}{Sections \ref{choice epic}~and~\ref{choice monic}}{\ref{choice epic}~and~\ref{choice monic}} that
$\choice$ preserves all epic morphisms and that it preserves all monic
morphisms $\varphi : X \to Y$ for which $Y$ is a discrete space.

\subsection{Terminal objects, products and limits of chains}
\label{app:limits}


We make use of three different instances of the notion of a limit:
terminal objects, products and limits of chains. We define these three
concepts separately in the following, but to the interested reader it
should be noted that they are all instances of the more general concept
of a limit for a diagram, which is discussed in all of the texts on
category theory cited above.

A \emph{terminal object} in a category $\catD$ is any object $\top$ with
the property that for every object $T$ of $\catD$ there exists a unique
morphism $\bang_T : T \to \top$. It follows directly from this
definition that any two terminal objects in some category need to be
isomorphic. To see this assume that we have a second terminal object
$\top'$. Because $\top'$ is terminal we have the morphism $\bang'_\top :
\top \to \top'$. Because $\top$ is universal there is the morphism
$\bang_{\top'} : \top' \to \top$. We have that $\id_\top = \bang_\top =
\bang_{\top'} \circ \bang'_\top$ because $\bang_\top$ is the unique
morphism from $\top$ to $\top$. A similar argument shows that
$\id_{\top'} = \bang'_\top \circ \bang_{\top'}$. As a consequence one
can assume that if a category has a terminal object then it is unique,
at least if one only cares about the identity of objects up-to
isomorphism.

The terminal object in the category $\Unc$ can be defined as the
uncertainty space $\top = (\{\star\},\{\emptyset,\{\star\}\})$ which
contains just one point. The terminal object in the category $\Unc^2$
can be defined componentwise as the object $(\top,\top)$ where $\top$ is
the terminal object of $\Unc$. A more interesting example of a terminal
object is the universal choice structure. The universality property from
Theorem~\ref{omega terminal} expresses precisely that the universal
choice structure is the terminal object in the category of choice
structures.

Given two objects $X$ and $Y$ of a category $\catD$, a \emph{product} of
$X$ and $Y$ is an object of $\catD$, written as $X \times Y$, together
with two morphisms $\pi_1 : X \times Y \to X$ and $\pi_2 : X \times Y \to
Y$, which have the following universal property: For every object $T$
and morphisms $f : T \to X$ and $g : T \to Y$ there is a morphism $u : T
\to X \times Y$, often written as $(f,g)$, that is unique for having the
property that $f = \pi_1 \circ u$ and $g = \pi_2 \circ u$.

Similar as for terminal objects, one can see that any two objects with
this universal property, relative to fixed $X$ and $Y$, need to be
isomorphic. Hence, one often speaks of the product of $X$ and $Y$ as if
it was unique.

In the category $\Unc$ one can define the product $X \times Y$ of two
uncertainty spaces $X$ and $Y$ to be an uncertainty space whose states
are all pairs $(x,y)$, where $x$ is a state of $X$ and $y$ is a state of
$Y$. The measurable sets of states are generated by taking finite unions
and complements of cylinders, that is, sets of the form $U \times Y$ and
$X \times V$ for measurable $U$ in $X$ and $V$ measurable in $Y$. It is
clear that with this algebra the projections $\pi_1 : X \times Y \to X,
(x,y) \mapsto x$ and $\pi_2 : X \times Y \to Y, (x,y) \mapsto y$ are
measurable functions. Moreover, one can check that this definition of
the product has the universal property that for each space $T$ together
with measurable functions $\varphi : T \to X$ and $\psi : T \to Y$ there
is a unique measurable function $\mu : T \to X \times Y, t \mapsto
(\varphi(t), \psi(t))$ such that $\varphi = \pi_1 \circ \mu$ and $\psi =
\pi_2 \circ \mu$.

We need a further piece of notation related to the product: Given $f : X
\to Z$ and $g : Y \to U$ we write $f \times g : X \times Y \to Z
\times U$ for the morphism $f \times g = (f \circ \pi_1,g \circ \pi_2)$.
It is easy to check that this definition generalizes our earlier
definition of $\varphi \times \psi$ for measurable functions $\varphi$ and
$\psi$ from Section~\ref{uncertainty spaces}.

A crucial categorical notion for the construction of the universal
choice structure is that of the limit of a countable cochain. A
\emph{countable cochain} $(X_n,f_n)_{n \in \omega}$ in a category
$\catD$ consists of an object $X_n$ of $\catD$ for every natural number
$n \in \omega$ and a morphism $f_n : X_{n + 1} \to X_n$ for every $n \in
\omega$. Here, and in the following we use $\omega$ to denote the set of
all natural numbers including $0$. The $f_n$ are also called the
\emph{coherence morphism} of the cochain.

It is clear that we can compose the morphisms $f_n$ to obtain a morphisms
$f^m_n = f_n \circ \dots \circ f_{m - 1} : X_m \to X_n$ for all $n,m \in
\omega$ with $n \leq m$, where $f^n_n = \id_{X_n} : X_n \to X_n$ is just
the identity on $X_n$.
We call the countable cochain $(X_n,f_n)_{n \in \omega}$ \emph{epic}, if
all the morphisms $f_n$ are epic.

A \emph{limit} of the countable cochain $(X_n,f_n)_{n \in
\omega}$ is an object $X_\omega$ together with projection morphisms $p_n
: X_\omega \to X_n$ such that $p_n = f_n \circ p_{n + 1}$ for all $n \in
\omega$ which satisfies the following universal property: For every
object $T$ together with morphisms $g_n : T \to X_n$, satisfying $g_n =
f_n \circ g_{n + 1}$ for all $n \in \omega$, there is a morphism $u : T
\to X_\omega$ that is unique for having the property that $g_n = p_n
\circ u$ for all $n \in \omega$. As for the product and the terminal
object, one can show with this universal property that any two limits
of a given cochain are isomorphic.

For every cochain $(X_n,\xi_n)_{n \in \omega}$ in the category $\Unc$ we
can define its limit $X_\omega$ to be the uncertainty space that has as
its states all sequences $x = (x_0,x_1,\dots,x_n,\dots)$, where $x_n$ is
a state from $X_n$ for all $n \in \omega$, which are coherent in the
sense that $x_n = \xi_n(x_{n + 1})$ for all $n \in \omega$. We can then
consider the projections $\zeta_n : X_\omega \to X_n, x \mapsto x_n$ for
each $n \in \omega$. The measurable sets of the limit $X_\omega$ are
defined to be all sets of the form $\zeta_n^{-1}[E]$ for some $n \in
\omega$ and measurable set $E$ of $X_n$. It is clear that this
definition turns the projections $\zeta_n : X_\omega \to X_n$ into
measurable functions. Moreover, using that all $\xi_n$ are measurable,
one can show that the measurable sets defined in this way are closed
under finite unions and intersections. To check that this definition of
the limit has the universal property observe that for any measurable
space $T$ with measurable functions $\varphi_n : T \to X_n$, such that
$\varphi_n = \xi_n \circ \varphi_{n + 1}$ for each $n \in \omega$, we
can define the unique morphism $\mu : T \to X_\omega$ such that $\mu(t)
= (\varphi_0(t), \varphi_1(t), \dots, \varphi_n(t),\dots)$ for all $t
\in T$.

One can also define the limit for every cochain
$((X_{0,n},X_{1,n}),(\xi_{0,n},\xi_{1,n}))_{n \in \omega}$ in $\Unc^2$.
This can be done componentwise, meaning that for each of the components
of the pairs representing objects and morphisms in $\Unc^2$ we can just
use the construction for $\Unc$, described in the previous paragraph.
For instance the limit itself is just the object
$(X_{0,\omega},X_{1,\omega})$ such that $X_{0,\omega}$ and
$X_{1,\omega}$ are the limits of $(X_{0,n},\xi_{0,n})_{n \in \omega}$
and $(X_{1,n},\xi_{1,n})_{n \in \omega}$ in $\Unc$. It is not hard to
see that this again has the required universal property.

A functor $G$ from $\catD$ to $\catE$ \emph{preserves} a limit
$X_\omega$ of a countable cochain $(X_n,f_n)_{n \in \omega}$ if $G
X_\omega$ is the limit for the countable cochain $(G X_n, G f_n)_{n \in
\omega}$.
The main technical result in this appendix is that the functor $\choice$
preserves limits of countable epic cochains. In the proof of this claim
in \iftoggle{preprint}{Section~\ref{choice preserves limits}}{\ref{choice preserves limits}} we investigate the
preservation of limits of cochains under the contravariant functors
$\acts$ from $\Unc$ to $\Set$ and $\justchoice$ from $\Set$ to $\Unc$.
For this purpose we need to understand limits of cochains in $\Set^\op$.
Because morphisms in $\Set^\op$ are just morphisms from $\Set$ with swapped domain and codomain we have
that a countable cochain $(X_n,f_n)_{n \in \omega}$ in $\Set^\op$ is
just what is called a countable chain in $\Set$, meaning that the $X_n$
are all sets and the $f_n : X_n \to X_{n + 1}$ are functions that now go
from $X_n$ to $X_{n + 1}$. Moreover the limit $X_\omega$ of this cochain
in $\Set^\op$ is what is actually called a colimit of chain in $\Set$.
The colimit $X_\omega$ of a chain $(X_n,f_n)_{n \in \omega}$ in $\Set$
can be described concretely as the disjoint union of all the $X_n$
modulo the equivalence relation that identifies an $x_n$ from $X_n$ with
an $x_m$ from $X_m$ if there is some $k \geq n,m$ such that $f^n_k(x_n)
= f^m_k(x_m)$. Clearly, we can then define inclusions $p^n_\omega : X_n
\to X_\omega$ for all $n \in \omega$. The universal property of this
colimit is just the same as the universal property of limit in
$\Set^\op$ with all morphisms turned around: For every set $T$ with
functions $h_n : X_n \to T$ such that $h_n = h_{n + 1} \circ f_n$ for
each $n \in \omega$, there is a unique function $u : X_\omega \to T$
such that $h_n = u \circ p^n_\omega$ for all $n \in \omega$.

\subsection{Natural transformations}
\label{app:nat trans}


Let $F$ and $G$ be functors that both go from a category $\catD$ to a
category $\catE$. Then a \emph{natural transformation} $\eta$ from
$F$ to $G$ is an assignment of a morphism $\eta_X : F X \to G X$ in
$\catE$ to every object $X$ in $\catD$ such that for all morphisms $f : X
\to Y$ in $\catD$ it holds that $G f \circ \eta_X = \eta_Y \circ F
f$, where $\eta_Y$ is the morphism that $\eta$ assigns to the object $Y$.

An example of a natural transformation from the discrete distribution
functor to the powerset functor is mapping a finite distribution to its
support set. More formally, we have for every set $X$ a function
$\supp_X : \dist X \to \powerset X$ that maps a discrete distribution
$\mu$ over $X$ to the support set $\supp_X(\mu)$ of $\mu$. The
naturality condition that $\powerset f \circ \supp_X = \supp_Y \circ
\dist f$ for all functions $f : X \to Y$ amounts to the fact that for
every discrete distribution $\mu$ over $X$ there is some $x \in X$ with
$f(x) = y$ and $\mu(x) \neq 0$ if and only if $\mu^f(y) \neq 0$, where
$\mu^f = \dist f (\mu)$.

An example of natural transformation from this paper is the measurable
function $\maxi_X : \prel X \to \choice X$ that is defined for every
uncertainty space $X$. Proposition~\ref{maxi natural}, which is proven
in \iftoggle{preprint}{Section~\ref{proof maxi natural}}{\ref{proof maxi natural}}, shows that $\maxi$ is a natural
transformation from the functor $\prel$ to the functor $\choice$.


\section{Properties of $\choice$}
\label{actual work}

In this part of the appendix we prove crucial properties of the functor
$\choice$ which are needed for the coalgebraic construction of the
universal choice structure. In case the reader tries to get an overview
of the construction of the universal choice structure, they might skip
this section and then refer back as needed when reading \iftoggle{preprint}{Sections
\ref{proofs for choices}~and~\ref{proofs for prel}}{\ref{proofs for choices}~and~\ref{proofs for prel}}.

\subsection{$\choice$ preserves epic morphisms}
\label{choice epic}

We show that $\choice$ preserves epic morphisms, which in $\Unc$ are just
surjective measurable functions. Because $\choice$ is the composition of
$\acts$ and $\justchoice$ it suffices to show that $\acts$ maps
surjective measurable functions to injective functions in $\Set$ and
that $\justchoice$ maps injective functions to surjective measurable
functions. It is easy to check the former using the definition of
$\acts$ and the cancellation property of epic morphism. That
$\justchoice$ maps injective functions to surjective measurable
functions is shown in the following lemma:

\begin{lemma} \label{justchoice monos to epics}
 If $f : X \to Y$ is injective then $\justchoice f : \justchoice Y \to
\justchoice X$ is surjective.
\end{lemma}
\begin{proof}
 It is easy to see that as $f$ is injective then $f^{-1}[f[U]] = U$
holds for all $U \subseteq X$. To check that $\justchoice f$ is
surjective consider any $C \in \justchoice X$. We have to find a $C' \in
\justchoice Y$ such that $\justchoice f (C') = C$. Define $C'$ such that
$C'(K) = f[C(f^{-1}[K])]$ for all finite $K \subseteq Y$. The following
equality then follows for all $L \subseteq X$:
\begin{align*}
 \justchoice f (C') (L) & = f^{-1}[C'(f[L])] \cap L \\
 & = f^{-1}[ f[     C(f^{-1}[f[L]])   ] ] \cap L \\
 & = C(L) \cap L = C(L).
\end{align*}
\end{proof}

\subsection{$\choice$ preserves monic morphisms to discrete spaces}
\label{choice monic}

We show that $\choice$ preserves monic morphisms, if the codomain has the
discrete algebra.
We again split the problem into two steps. First we show that $\acts$
maps every surjective measurable function with a discrete codomain to an
injective functions and then we show that $\justchoice$ maps injective
functions to surjective measurable functions.

\begin{lemma}
 If the measurable function $\varphi : X \to Y$ is injective and $Y$ has
the discrete algebra then the function $\acts \varphi : \acts Y \to
\acts X$ is surjective.
\end{lemma}
\begin{proof}
 Pick any $f \in \acts X$. We want to find a $f' \in \acts Y$ such that
$f = f' \circ \varphi$. We define $f'(y) = f(x)$, if there is some $x \in
X$ such that $\varphi(x) = y$, and $f'(y) = z'$ for an arbitrary $z' \in
Z$ otherwise. This is well-defined because $\varphi$ is injective. The
function $f'$ is trivially measurable because $Y$ has the discrete
algebra and $f = f' \circ \varphi$ by definition of $f'$.
\end{proof}

\begin{lemma}
 If $f : X \to Y$ is surjective then $\justchoice f : \justchoice Y \to
\justchoice X$ is injective.
\end{lemma}
\begin{proof}
 It is easy to see that as $f$ is surjective it holds for all $U
\subseteq Y$ that $f[f^{-1}[U]] = U$.
To check that $\justchoice f$ is injective consider any two $C,C' \in
\justchoice Y$ such that $\justchoice f(C) = \justchoice f(C')$. We have
to show that then $C(K) = C'(K)$ for all finite $K \subseteq Y$.
Using that $f[f^{-1}[U]] = U$, we can compute:
\begin{align*}
 \justchoice f (C) (f^{-1}[K]) & = f^{-1}[C(f[f^{-1}[K]])] \cap
f^{-1}[K] \\
 & = f^{-1}[C(K)] \cap f^{-1}[K] = f^{-1}[C(K) \cap K] = f^{-1}[C(K)]
\end{align*}
Similarly we obtain also for $C'$ that $\justchoice f(C') (f^{-1}[K]) =
f^{-1}[C'(K)]$. From the assumption that $\justchoice f(C) = \justchoice
f(C')$ it follows that $f^{-1}[C(K)] = f^{-1}[C'(K)]$. Because $f$ is
surjective it follows that $C(K) = C'(K)$.
\end{proof}

\subsection{$\choice$ preserves limits of countable epic cochains}
\label{choice preserves limits}



In this section we prove the main technical result of this paper:
\begin{theorem} \label{preservation of limits}
 $\choice$ preserves limits of countable epic cochains. This means that
whenever we have a limit $X_\omega$, with projections $\zeta_n :
X_\omega \to X_n$ for all $n \in \omega$, of a countable cochain
$(X_n,\xi_n)_{n \in \omega}$ in which the coherence morphisms $\xi_n$
are epic for all $n$, then $\choice X_\omega$, with projections $\choice
\zeta_n : \choice X_\omega \to \choice X_n$, satisfies the universal
property of the limit of the cochain $(\choice X_n, \choice \xi_n)_{n
\in \omega}$, i.e., for every uncertainty space $T$ together with
measurable functions $\varphi_n : T \to \choice X_n$ for all $n \in
\omega$ with $\varphi_n = \choice \xi_n \circ \varphi_{n + 1}$ there is
a unique measurable function $\mu : T \to \choice X_\omega$ such that
$\varphi_n = \choice \zeta_n \circ \mu$ for all $n \in \omega$.
\end{theorem}

For the proofs from later sections in this appendix we need the
following corollary of Theorem~\ref{preservation of limits}.
\begin{corollary} \label{choicePlayers preserves limits}
 $\choicePlayers$ preserves limits of countable epic cochains.
\end{corollary}

We split the proof of Theorem~\ref{preservation of limits} into two
steps. First, we show that $\acts$ turns limits of epic cochains in
$\Unc$ into colimits of chains in $\Set$. Second, we show that
$\justchoice$ turns colimits of chains $\Set$ into limits of cochains in
$\Unc$. It follows that $\choice$ preserves limits of epic cochains
because $\choice$ is defined as the composition of $\acts$ and
$\justchoice$.

To prove the preservation result about $\acts$ we need two technical lemmas. The first relates acts and measurable functions. It
crucially relies on our assumptions that acts are finite step functions.
\begin{lemma} \label{factorisation}
 Let $\varphi : X \to Y$ be a measurable function and $f \in \acts X$ an
act with range $f[X] = \{z_1,\dots,z_k\}$ such that for every $n \in
\{1,\dots,k\}$ there is some measurable set $E_n$ in $Y$ such that
$f^{-1}[\{z_n\}] = \varphi^{-1}[E_n]$. Then there is an act $f' \in
\acts Y$ such that $f = f' \circ \varphi$.
\end{lemma}
\begin{proof}
 We use an induction to define $H_1 = E_1$ and $H_{n + 1} = E_{n + 1}
\setminus \bigcup_{m = 1}^n H_m$ for $n \in \{1,\dots,k - 1\}$. We also
define $H_0 = Y \setminus \bigcup_{m = 1}^k H_k$. It is clear that all
these sets $H_0,H_1,\dots,H_k$ are measurable in $Y$. Moreover, it is
also obvious from the definition that they partition the set $Y$. Let
$z_0 \in Z$ be an arbitrary element of the non-empty set of outcomes.
Define the act $f' \in \acts Y$ such that $f'(y) = z_n$ if $y \in H_n$.
Because the sets $H_0,H_1,\dots,H_k$ partition $Y$ this is well-defined
as a function. It is a measurable function because the sets
$H_0,H_1,\dots,H_k$ are measurable. To see that $f = f' \circ \varphi$,
fix any $x \in X$ and let $n$ be such that $f(x) = z_n$. Because
$f^{-1}[\{z_n\}] = \varphi^{-1}[E_n]$ we then have that $x \in
\varphi^{-1}[E_n]$. Also, for every $m \neq n$ the sets
$f^{-1}[\{z_n\}]$ and $f^{-1}[\{z_m\}] = \varphi^{-1}[E_m]$ are disjoint
and hence $x \notin \varphi^{-1}[E_m]$. It follows that $\varphi(x) \in
H_n$ and thus $f'(\varphi(x)) = z_n = f(x)$.
\end{proof}

The next lemma states a property that also plays an important role in
Proposition~1 of \cite{DiTillio08}:
\begin{lemma} \label{pushing acts down}
Consider a countable cochain $(X_n,\xi_n)_{n \in \omega}$, with
coherence morphisms $\xi_n : X_{n + 1} \to X_n$, and let $X_\omega$,
with projections $\zeta_n : X_\omega \to X_n$, be its limit. Then, for
every act $f \in \acts X_\omega$ there is a natural number $n \in
\omega$ such that for every natural number $m \geq n$ there is an act
$f' \in \acts X_m$ such that $f = f' \circ \zeta_m$.
\end{lemma}
\begin{proof}
 Because $f : X_\omega \to Z$ is a finite step function we can enumerate
its range as $f[X] = \{z_1,\dots,z_k\}$. Then observe that for every $i
\in \{1,\dots,k\}$ the set $f^{-1}[\{z_i\}]$ is measurable in
$X_\omega$, since $f$ is a measurable and $Z$ has the discrete algebra.
By the definition of the algebra on $X_\omega$ this means that for every
$i \in \{1,\dots,k\}$ there is some $n_i \in \omega$ such that
$f^{-1}[\{z_i\}] = \zeta_{n_i}^{-1}[E_i]$ for some measurable set $E_i$
in $X_{n_i}$. Let $n$ be the maximum of all $n_i$ for $i \in
\{1,\dots,k\}$.

Consider then any $m \geq n$ and for all $i \in \{1,\dots,k\}$ define $E^m_i = (\xi^m_{n_i})^{-1}[E_i]
\subseteq X_m$. These $E^m_i$ are
measurable in $X_m$ because $\xi^m_{n_i}$ is a measurable function. By
the definition of the limit of a cochain we also have that $\zeta_{n_i} = \xi^m_{n_i} \circ \zeta_m$ for every $i
\in \{1,\dots,k\}$. Hence $f^{-1}[\{z_i\}] = \zeta_{n_i}^{-1}[E_i] =
(\xi^m_{n_i} \circ \zeta_m)^{-1}[E_i] =
\zeta_m^{-1}[(\xi^m_{n_i})^{-1}[E_i]] = \eta_m^{-1}[E^m_i]$. It thus
follows with Lemma~\ref{factorisation} that there is an act $f' \in X_m$
such that $f = f' \circ \zeta_m$.
\end{proof}


\begin{proposition}
 $\acts$ preserves limits of countable epic cochains. That is, it
maps limits of cochains of uncertainty spaces to colimits of chains of
sets.
\end{proposition}
\begin{proof}
 Consider a cochain $(X_n,\xi_n)_{n \in \omega}$ of uncertainty spaces
and let $X_\omega$, together with projections $\zeta_n : X_\omega \to
X_n$, be its limit. We show that $\acts X_\omega$, together with the
inclusions $\acts \zeta_n : \acts X_n \to \acts X_\omega$ has the universal
property of the colimit of the chain $(\acts X_n,\acts \xi_n)_{n \in
\omega}$. For this purpose consider any set $T$ together with functions
$h_n : X_n \to T$ such that $h_{n} = h_{n+1} \circ \acts \xi_n$
for each $n \in \omega$. We need to define the function $u : \acts
X_\omega \to T$ such that $h_n = u \circ \acts \zeta_n$ for all $n \in \omega$
and show that it is unique with that property.

To define $u(f)$ for some act $f \in \acts X_\omega$ we use
Lemma~\ref{pushing acts down}. From this lemma it follows that there is
some $k$ and act $f' \in \acts X_k$ such that $f = f' \circ \zeta_k = \acts
\zeta_k (f')$. We then set $u(f) = h_k(f')$. To check that this satisfies
$h_n (g) = u \circ \acts \zeta_n (g)$ for all $n \in \omega$ and $g \in
\acts X_n$ consider the act $\acts \zeta_n(g) \in \acts X_\omega$. By
definition of $u$ we have that $u(\acts \zeta_n (g)) = h_k(f')$ for some
$f' \in \acts X_k$ such that $\acts \zeta_n(g) = \acts \zeta_k (f')$. We
need to show that $h_k(f') = h_n(g)$. Assume that $n \leq k$. This is
without loss of generality because we are only using the completely
symmetric fact that $\acts \zeta_n(g) = \acts \zeta_k (f')$. Because $\zeta_n
= \xi^k_n \circ \zeta_k$, it follows from $\acts \zeta_n(g) = \acts \zeta_k
(f')$ that $f' \circ \zeta_k = g \circ \zeta_n = g \circ \xi^k_n \circ
\zeta_k$. We obtain that $f' = g \circ \xi^k_n = \acts \xi^k_n (g)$
because $\zeta_k$ is epic. From this we can then conclude that $h_k(f') =
h_n(g)$ since $h_k = h_n \circ \acts \xi^k_n$.

That $u$ is unique also follows from Lemma~\ref{pushing acts down}. The
possible values of $u(f)$ are completely determined because $f = \acts
\zeta_n(f')$ for some $n$ and $f' \in X_n$ and we need to ensure that
$h_n(f') = u \circ \acts \zeta_n (f')$.
\end{proof}

\begin{lemma} \label{finite sets at finite level}
 Consider a chain $(X_n, \xi_n)_{n \in \omega}$ of sets and let
$X_\omega$, with inclusions $\iota_n : X_n \to X_\omega$ for all $n$, be
its colimit. Then for each finite $K \subseteq X_\omega$ there is an $m
\in \omega$ and a $K' \subseteq X_m$ such that $\iota_m[K'] = K$.
\end{lemma}
\begin{proof}
 By the definition of the colimit of a countable chain we have that for every
$k \in X_\omega$ there exists some $k' \in X_{m_k}$ such that
$\iota_{m_k}(k') = k$. Let $m$ be the maximum of the finitely many
$m_k$ for all $k \in K$ and let then $K' = \{\xi^{m_k}_m(k') \in X_m
\mid k \in K\}$. It then holds that $\iota_m[K'] = K$ because $\iota_m
\circ \xi^{m_k}_m(k') = \iota_{m_k}(k') = k$ for every $k \in K$.
\end{proof}

\begin{theorem} \label{justchoice preserves limits}
 $\justchoice$ preserves colimits of countable chains. That is, it maps
colimits of chains of sets onto limits of cochains of uncertainty
spaces.
\end{theorem}
\begin{proof}
 Consider a chain $(X_n, \xi_n)_{n \in \omega}$ of sets and let
$X_\omega$, with inclusions $\iota_n : X_n \to X_\omega$ for all $n$, be
its colimit. We show that $\justchoice X_\omega$, with the
$\justchoice \iota_n : \justchoice X_\omega \to \justchoice X_n$ as the
projections, has the universal property of the limit of the cochain
$(\justchoice X_n, \justchoice \xi_n)_{n \in \omega}$. Hence, consider
any uncertainty space $T$ together with measurable functions $\varphi_n
: T \to \justchoice X_n$, satisfying $\varphi_n = \justchoice \xi_n
\circ \varphi_{n + 1}$, for all $n \in \omega$. We need to show that
there is a unique measurable function $\mu : T \to \justchoice X_\omega$
such that $\justchoice \iota_n \circ \mu = \varphi_n$ for all $n \in
\omega$.

We first describe how to define the choice function $\mu(t) \in
\justchoice X_\omega$ on a finite set $K \subseteq X_\omega$. Let $m \in
\omega$ be the least number such that there exists a $K' \subseteq X_m$
such that $\iota_m[K'] = K$. Such a number exists because of
Lemma~\ref{finite sets at finite level}. Also observe that the choice of
$m$ and $K'$ only depends on $K$ but not on $t$. This is a property
which we exploit below to show that $\mu$ is measurable. We then set the
value of the choice function $\mu(t)$ on $K$ to be
\begin{equation} \label{def mu}
 \mu(t)(K) = \iota_m[\varphi_m(t)(K')],
\end{equation}
where $\varphi_m(t)(K') \subseteq K'$ denotes the elements of $K'$
selected by the choice function $\varphi_m(t) \in \justchoice X_m$. This
is well-defined because we have $\iota_m[\varphi_m(t)(K')] \subseteq
\iota_m[K'] = K$

To see that $\mu$ is measurable it suffices to check that the preimage
$\mu^{-1}[B^K_L] \subseteq T$ of a basic measurable set $B^K_L = \{C \in
\justchoice X_\omega \mid C(K) \subseteq L\}$ is measurable in $T$.
Hence, fix finite $K$ and $L$ with $L \subseteq K$ and let $m \in
\omega$ and $K'$ be defined from $K$ as explained above. We are now
going to show that $\mu^{-1}[B^K_L] = \varphi_m^{-1}[B^{K'}_{L'}]$,
where $L' = \iota_m^{-1}[L]$. Because $\varphi_m : T \to \choice X_m$ is
assumed to be measurable, it then follows that $\mu^{-1}[B^K_L]$ is
measurable too. By unfolding the definitions one sees that the claim
that $\mu^{-1}[B^K_L] = \varphi_m^{-1}[B^{K'}_{L'}]$ is equivalent to
the claim that for all $t \in T$
\[
 \mu(t)(K) \subseteq L \iff \varphi_m(t)(K') \subseteq \iota_m^{-1}[L] .
\]
By the definition of $\mu(t)$ above we see that the left side of this
equivalence is the same as $\iota_m[\varphi_m(t)(K')] \subseteq L$,
which is clearly equivalent to $\varphi_m(t)(K') \subseteq
\iota_m^{-1}[L]$.

We need to verify that $\varphi_n = \justchoice \iota_n \circ \mu$ for
all $n \in \omega$. This requires that for every $t \in T$ the choice
functions $\varphi_n(t)$ and $\justchoice \iota_n \circ \mu(t)$ are the
same. Hence, they need to have the same value on every finite set $L'
\subseteq X_n$.

First let $K = \iota_n[L']$ and recall the definitions to see that
\begin{equation} \label{first unfolding}
 (\justchoice \iota_n \circ \mu(t))(L') =
\iota_n^{-1}[\mu(t)(\iota_n[L'])] \cap L' =
\iota_n^{-1}[\iota_m[\varphi_m(t)(K')]] \cap L',
\end{equation}
for the $m \in \omega$ and the finite set $K' \subseteq X_m$ that are
defined from $K$ as described above. Our choice of $K'$ above ensures
that $\iota_m[K'] = K = \iota_n[L']$. Hence for every $k' \in K'$ there
is some $l' \in L'$ such that $\iota_m(k') = \iota_n(l')$ and conversely
for every $l' \in L'$ there is a $k' \in K'$ such that $\iota_m(k') =
\iota_n(l')$. By the definition of identity for elements in the colimit
$X_\omega$ it follows that for each such pair $k' \in X_m$ and $l' \in
X_n$ with $\iota_m(k') = \iota_n(l')$ there is some $j \in \omega$ such
that $\xi^m_j(k') = \xi^n_j(l')$. Let $i$ be the maximum of all those $j
\in \omega$, which exists because there are finitely many pairs $(k',l')
\in K' \times L'$. It then clearly holds that $\xi^m(k') =
\xi^n(l')$, whenever $\iota_m(k') = \iota_n(l')$ for $(k',l') \in K'
\times L'$. Define then $L \subseteq X$ to be the finite set
\[
 L = \xi^m[K'] = \xi^n[L'].
\]
Because $\varphi_n =  \justchoice \xi^n \circ \varphi$ we obtain
that
\[
 \varphi_n (t) (L') = (\xi^n)^{-1}[\varphi(t)(\xi^n_i[L'])] \cap L'
= (\xi^n)^{-1}[\varphi(t)(L)] \cap L' .
\]
Similarly, because $\varphi_m =  \justchoice \xi^m \circ \varphi$ we
obtain
\[
 \varphi_m (t) (K') = (\xi^m)^{-1}[\varphi_m(t)(\xi^m[K'])] \cap K'
= (\xi^m)^{-1}[\varphi(t)(L)] \cap K' .
\]
To prove $\varphi_n (t) (L') = (\justchoice \iota_n \circ \mu(t))(L')$
it thus suffices by \eqref{first unfolding} to show that
\begin{equation} \label{final step}
 (\xi^n)^{-1}[U] \cap L' = \iota_n^{-1}[\iota_m[(\xi^m)^{-1}[U] \cap
K']] \cap L'
\end{equation}
for the set $U = \varphi(t)(L) \subseteq L$.

For the left-to-right inclusion of \eqref{final step} consider any $l'
\in L'$ such that $\xi^n(l') \in U$. We need that $l' \in
\iota_n^{-1}[\iota_m[(\xi^m)^{-1}[U] \cap K']]$. From the definition of
$L$ we get that there is then some $k' \in K'$ such that $\xi^m(k') =
\xi^n(l')$. Hence $k' \in (\xi^m)^{-1}[U] \cap K'$ and $\iota_n(l')
= \iota_m(k')$. The latter two directly entail that $l' \in
\iota_n^{-1}[\iota_m[(\xi^m)^{-1}[U] \cap K']]$.

For the right-to-left inclusion of \eqref{final step} consider any $l'
\in L'$ such that $\iota_n(l') = \iota_m(k')$ for some $k' \in
(\xi^m)^{-1}[U] \cap K'$. With the definition of $i$ above it then
follows from $\iota_n(l') = \iota_m(k')$ that $\xi^n(l') =
\xi^m(k') \in U$. Hence, $l' \in (\xi^m)^{-1}[U]$ and we are done.

Lastly, we show that $\mu(t)$ is completely
determined by the requirement that $\varphi_n(t) = \justchoice \iota_n
\circ \mu(t)$ for all $n \in \omega$. Consider any $\nu : T \to
\justchoice X_\omega$ such that $\varphi_n(t) = \justchoice \iota_n
\circ \nu(t)$ for all $n \in \omega$. We argue that then $\nu(t)(K) =
\mu(t)(K)$ for the $\mu$ defined as above and all finite $K \subseteq
X_\omega$. Let $m \in \omega$ and $K' \subseteq X_m$ as in the
definition of $\mu$ above, i.e., such that $\iota_m[K'] = K$.
The requirement that $\varphi_m(t) = \justchoice \iota_m \circ \nu(t)$
means that
\[
 \varphi_m(t)(K') = \iota_m^{-1}[\nu(t)(\iota_m[K'])] \cap K' =
\iota_m^{-1}[\nu(t)(K)] \cap K'. \]
One can see that this entails that
\[
 \iota_m[\varphi_m(t)(K')] = \nu(t)(K).
\]
The left-to-right inclusion is trivial and the other inclusion follows
because $\nu(t)(K) \subseteq K$ and $\iota_m[K'] = K$. Therefore, we have
shown that $\nu(t)(K)$ equals the expression that we use in \eqref{def
mu} to define
$\mu(t)(K)$.
\end{proof}

As an immediate consequence of the previous two propositions we obtain
the main result of this section.

\section{Proofs for Section~\ref{choice structures}}
\label{proofs for choices}

Theorems \ref{iso at omega}~and~\ref{omega terminal} follow from the
results about $\choice$ together with well-known results \cite{Worrell99}
about the terminal coalgebra. In the following, we sketch how these results
from the general theory of coalgebras apply to the setting of choice
structures.


\subsection{Preliminary observations}
\label{preliminary observations}


Theorems \ref{iso at omega}~and~\ref{omega terminal} concern the
universal choice structure that is obtained from the choice hierarchies
as described in Section~\ref{universal choice structure}. Let us first
see how this fits into the coalgebraic set-up.

One can view the two uncertainty spaces $\Omega_{i,n}$ and
$\Omega_{j,n}$ on the $n$-th level in the choice hierarchy as defining an
object $\Omega_n = (\Omega_{i,n},\Omega_{j,n})$ in the category
$\Unc^2$. In fact one can define these objects directly, just using the
functor $\choicePlayers$ as follows: We start with $\Omega_0 = \top$, where
$\top$ is the terminal object in $\Unc^2$, and then define inductively
$\Omega_{n + 1} = \choicePlayers \Omega_n$. It is easy to see that with the exception of
the $0$-th level, which is omitted from the discussion in the main text,
this yields the same sequence of pairs of uncertainty spaces as defined
in Section~\ref{universal choice structure}.
Similarly, the coherence morphism $\xi_{i,n}$ and $\xi_{j,n}$ can also be
defined directly in $\Unc^2$ with $\xi_0 = \bang_{\choicePlayers \top} :
\choicePlayers \top \to \top$ and inductively $\xi_{n + 1} = \choicePlayers
\xi_n : \choicePlayers \Omega_{n + 1} \to \choicePlayers \Omega_n$.
Thus, the sequence $(\Omega_n,\xi_n)_{n \in \Omega}$ forms a countable
cochain in $\Unc^2$.

The definition of the uncertainty spaces $\Omega_i$ and $\Omega_j$ of
types in the universal choice structures from Section~\ref{universal
choice structure} is such that these are precisely the limits of the
countable cochains $(\Omega_{i,n},\xi_{i,n})_{n \in \omega}$ and
$(\Omega_{j,n},\xi_{j,n})_{n \in \omega}$ in $\Unc$.
Because limits in $\Unc^2$ are computed component-wise it follows that
the object $\Omega = (\Omega_i,\Omega_j)$ in $\Unc^2$ is the limit of the countable cochain $(\Omega_n,\xi_n)_{n \in \omega}$. Also recall that
the limit comes with projections $\zeta_n = (\zeta_{i,n},\zeta_{j,n}) :
\Omega \to \Omega_n$ back into the chain such that $\zeta_n = \xi_n \circ \zeta_{n + 1}$.




The crucial observation behind Theorems \ref{iso at
omega}~and~\ref{omega terminal} is then that the functor
$\choicePlayers$ on $\Unc^2$ preserves the limit of the epic cochain
$(\Omega_n,\xi_n)$. This follows from Corollary~\ref{choicePlayers
preserves limits}, which states that $\choicePlayers$ preserves limits
of countable epic cochains. To apply Corollary~\ref{choicePlayers
preserves limits} to the cochain $(\Omega_n,\xi_n)_{n \in \omega}$ we
need to check that the coherence morphisms $\xi_n$ for all $n \in
\omega$ are epic. This can be checked directly by an induction over the
definition of the $\xi_n$. In the base step $\xi_0 =
(\xi_{i,0},\xi_{j,0})$, and $\xi_{i,0} = \choice \pi_1$, where $\pi_1$
is the projection out of a product. Because this projection is epic, and
by the argument in \iftoggle{preprint}{Section~\ref{choice epic}}{\ref{choice epic}} the functor $\choice$ preserves epic
morphisms, it follows that $\xi_{i,0}$ is epic. We reason analogously
for $\xi_{j,0}$. That $\xi_{n + 1} = (\xi_{i,n + 1},\xi_{j,n + 1}) =
(\choice (\id_{A_j} \times \xi_{j,n}), \choice (\id_{A_i} \times
\xi_{i,n}))$ is epic, also follows easily because $\choice$ preserves
epic morphisms.

\subsection{Theorem~\ref{iso at omega}}
\label{definition mu}

Next consider the object $\choicePlayers \Omega$ of $\Unc^2$. For this
object we can define morphisms into the cochain $(\Omega_n,\xi_n)_{n \in
\omega}$ by setting $\tau_0 = \bang_{\choicePlayers \Omega} :
\choicePlayers \Omega \to \Omega_0$ and $\tau_{n + 1} = \choicePlayers
\zeta_n : \choicePlayers \Omega \to \choicePlayers \Omega_n$, which
satisfy $\tau_n = \xi_n \circ \tau_{n + 1}$ for all $n$. From
Theorem~\ref{preservation of limits} it follows that $\choicePlayers
\Omega$ together with the projections $\choicePlayers \zeta_n = \tau_{n
+ 1}$ is a colimit of the cochain $(\choicePlayers
\Omega_n,\choicePlayers \xi_n)_{n \in \omega}$, which by definition is
the same as the chain $(\Omega_{n + 1},\xi_{n + 1})_{n \in \omega}$. We
can use this observation to show that $\choicePlayers \Omega$ together
with the $\tau_n$ also satisfies the universal property of the limit of
the cochain $(\Omega_n,\xi_n)_{n \in \omega}$. To this aim take any
further object $T$ of $\Unc^2$ together with morphisms $g_n : T \to
\Omega_n$ such that $g_n = \xi_n \circ g_{n + 1}$ for all $n \in
\omega$. We now just consider this $g_n$ without $g_0$ as morphism into
the chain $(\choicePlayers \Omega_n,\choicePlayers \xi_n)_{n \in
\omega}$ satisfying that $g_{n + 1} = \xi_{n + 1} \circ g_{n + 2} =
\choicePlayers \xi_n \circ g_{n + 1}$ for all $n \in \omega$. Because
$\choicePlayers \Omega$ is a limit of this chain there must be then a
unique morphism $u : T \to \choicePlayers \Omega$ such that $g_{n + 1} =
\choicePlayers \zeta_n \circ u = \tau_{n + 1} \circ u$ for all $n \in
\omega$. Additionally, we have that $g_0 = \tau_0 \circ u$ holds
trivially because on both sides of the equation we have a morphism into
the terminal object $\top = \Omega_0$ of $\Unc^2$.

We have now seen that both $\Omega$ and $\choicePlayers \Omega$ satisfy the
universal property of limit for the cochain $(\Omega_n,\xi_n)_{n \in
\omega}$. Theorem~\ref{iso at omega} then follows because there can be
only one such object up to isomorphism.

More precisely, the isomorphism is given by the unique morphism
$\mu = (\mu_i,\mu_j) : \Omega \to \choicePlayers \Omega$ such that $\zeta_n =
\tau_n \circ \mu$ for all $n \in \omega$, which exists because $\Omega$
is a limit. If one considers the components of this morphism one obtains
the isomorphism $\mu_i : \Omega_i \to \choice(A_j \times \Omega_j)$ and
$\mu_j : \Omega_j \to \choice(A_i \times \Omega_i)$ that are referred to in
the formulation of Theorem~\ref{iso at omega}.

\begin{remark}
 One can generate a topology on $\Omega_i$, using the algebra of events as the basis.
 Because of our assumption that $Z$, $A_i$ and $A_j$ are finite it follows that all the $\Omega_{i,n}$ are finite. Thus, $\Omega_i$ is what is called a filtered limit of finite spaces, which is well-know to be a Stone space, meaning that it is compact, Hausdorff and has a basis of clopen sets. One can also show also that this topology has a countable basis, because for every of the countably many spaces $\Omega_{i,n}$ there are only finitely many events in $\Omega_i$ that are preimages of sets in $\Omega_{i,n}$. Compact Hausdorff spaces with a countable basis are called simple spaces in \cite{DiTillio08}.
\end{remark}

\begin{remark} \label{r:finiteness}
 For the result of this section, and the universality of the universal choice structure proved in the next section, it is not needed that the uncertainty spaces $Z$, $A_i$, and $A_j$ are finite. The preservation of limits of epic cochains, which is at the core of these results, holds for arbitrary uncertainty spaces. However, for our proof of Lemma~\ref{factorisation} we need that acts are defined as finite step functions and that $Z$ carries the discrete algebra. The finiteness of $Z, A_i$ and $A_j$ is required for the results about non-redundancy from Section~\ref{nonred} and the embedding of the universal preference structure from Section~\ref{embedding}.
\end{remark}



\subsection{Theorem~\ref{omega terminal}}
\label{why terminal}

We now prove Theorem~\ref{omega terminal}. Consider an
arbitrary choice structure $\X = (T_i,T_j,\theta_i,\theta_j)$, presented
as a coalgebra $(T,\theta) = ((T_i,T_j),(\theta_i,\theta_j))$ for
$\choicePlayers$.
We need to show that there is a unique morphism $\upsilon$ from
$(T,\theta)$ to the universal choice structure $(\Omega,\mu)$.

To prove the existence of $\upsilon$ first consider the measurable
functions $\upsilon_0 = \bang_T : T \to \Omega_0$ and the inductively
defined $\upsilon_{n + 1} = \choicePlayers \upsilon_n \circ \theta : T
\to \Omega_{n + 1}$. By an induction over $n$ we can show that these
measurable functions satisfy $\upsilon_n = \xi_n \circ \upsilon_{n +
1}$. In the base case this is clear because there is only one
morphism from $T$ to the terminal object $\Omega_0 = \top$. For the
inductive step we calculate as follows:
\begin{align*}
 \upsilon_{n + 1} & = \choicePlayers \upsilon_n \circ \theta =
\choicePlayers (\xi_n \circ \upsilon_{n + 1}) = \choicePlayers \xi_n
\circ \choicePlayers \upsilon_{n + 1} \circ \theta
= \xi_{n + 1} \circ \upsilon_{n + 2}.
\end{align*}
Because $\Omega$ is defined as the limit of the sequence
$(\Omega_n,\xi_n)_{n \in \Omega}$ it follows that there is a unique
morphism $\upsilon : T \to \Omega$ with the property that for all $n \in
\omega$
\begin{equation} \label{fancy upsilon}
 \upsilon_n = \zeta_n \circ \upsilon .
\end{equation}

It remains to be seen that $\upsilon$ is a morphism of choice structures
and that it is unique with this property.

To show that $\upsilon : T \to \Omega$ is a morphism from $(T,\theta)$
to $(\Omega,\mu)$ we need to verify that $\mu \circ \upsilon =
\choicePlayers \upsilon \circ \theta$. We show this by proving that
$\upsilon_n = \tau_n \circ \mu \circ \upsilon$ and $\upsilon_n = \tau_n
\circ \choicePlayers \upsilon \circ \theta$ hold for all $n \in \omega$.
The equality $\mu \circ \upsilon = \choicePlayers \upsilon \circ \theta$
then follows from the universal property for the limit $\choicePlayers
\Omega$ of $(\Omega_n,\xi_n)_{n \in \omega}$, with projections $\tau_n :
\choicePlayers \Omega \to \Omega_n$.

To see that $\upsilon_n = \tau_n \circ \mu \circ \upsilon$ we use that
by its definition in \iftoggle{preprint}{Section~\ref{definition mu}}{\ref{definition mu}} $\mu$ satisfies $\tau_n
\circ \mu = \zeta_n$. The claim then reduces to \eqref{fancy upsilon}.

To see that $\upsilon_n = \tau_n \circ \choicePlayers \upsilon \circ
\theta$ we distinguish two cases. If $n = 0$ we
have that both $\upsilon_0$ and $\tau_0 \circ \choicePlayers \upsilon \circ
\theta$ are morphisms from $T$ to the terminal object $\Omega_0 = \top$ of
$\Unc^2$. By the universal property of the terminal object they must be
equal. 

If $n$ is strictly positive we can unfold the definition of $\upsilon_{n
+ 1}$, use \eqref{fancy upsilon} and then unfold the definition of
$\tau_{n + 1}$ to obtain the following computation
\[
 \upsilon_{n + 1} = \choicePlayers \upsilon_n \circ \theta =
\choicePlayers (\zeta_n \circ \upsilon) \circ \theta = \choicePlayers
\zeta_n \circ \choicePlayers \upsilon \circ \theta = \tau_{n + 1} \circ
\choicePlayers \upsilon \circ \theta.
\]

To prove that $\upsilon$ is the only morphism of choice structures from
$(T,\theta)$ to $(\Omega,\mu)$ consider any other morphism $\upsilon' :
T \to \Omega$ in $\Unc^2$ such that $\mu \circ \upsilon' =
\choicePlayers \upsilon' \circ \theta$. We show that then $\upsilon_n =
\zeta_n \circ \upsilon'$ for all $n \in \omega$, from which is follows
that $\upsilon' = \upsilon$ because $\upsilon$ is defined as the unique
morphism with the property \eqref{fancy upsilon}. To prove
$\upsilon_n = \zeta_n \circ \upsilon'$ we use an induction on $n$. In
the base case we again have that $\upsilon_0$ and $\zeta_0 \circ
\upsilon'$ must be equal because they are both morphism from $T$ to the
terminal object $\Omega_0 = \top$ of $\Unc^2$.
In the inductive step we use the following computation:
\begin{align*}
 \upsilon_{n + 1} & = \choicePlayers \upsilon_n \circ \theta &
\mbox{definition of } \upsilon_{n + 1} \\
& = \choicePlayers (\zeta_n \circ \upsilon') \circ \theta &
\mbox{induction hypothesis} \\
& = \choicePlayers \zeta_n \circ \choicePlayers \upsilon' \circ \theta &
\choicePlayers \mbox{ functor} \\
& = \tau_{n + 1} \circ \choicePlayers \upsilon' \circ \theta &
\mbox{definition of } \tau_{n + 1} \\
& = \tau_{n + 1} \circ \mu \circ \upsilon' & \mbox{assumption on }
\upsilon' \\
& = \zeta_{n + 1} \circ \upsilon' & \mbox{uniqueness property of }
\mu
\end{align*}

\subsection{Results for Section~\ref{nonred}}
\label{app:chara non-red}

In this part of the appendix we prove our results about non-redundancy.

\subsubsection{Well-definedness of Definition~\ref{observables}}
\label{app:observables well defined}


To explain the definition of the algebra of observable events more
extensively, fix the choice structure $\X = (T_i,T_j,\theta_i,\theta_j)$
and consider pairs of algebras $(\alg_i,\alg_j)$ such that $\alg_i$ is an
algebra over the set $T_i$ and $\alg_j$ is an algebra over the set
$T_j$. Define an order $\leq$ over such pairs of algebras such that
$(\alg_i,\alg_j) \leq (\alg'_i,\alg'_j)$ iff $\alg_i \subseteq \alg'_i$
and $\alg_j \subseteq \alg'_j$. Definition~\ref{observables} then
defines $(\alg^{ob}_{\X,i},\alg^{ob}_{\X,j})$ as the least element
$(\alg_i,\alg_j)$ in this order $\leq$ which satisfies that
\begin{enumerate}
 \item $B^{ob}_{\X,i}(K,F) \in \alg_i$ for all finite $K,L \subseteq
\acts(A_j \times (T_j,\alg_j))$, and
 \item $B^{ob}_{\X,j}(K,F) \in \alg_i$ for all finite $K,L \subseteq
\acts(A_i \times (T_i,\alg_i))$.
\end{enumerate}
Let $O$ be the set of all pairs of algebras that satisfy these two
conditions. The well-definedness of
$(\alg^{ob}_{\X,i},\alg^{ob}_{\X,j})$ now amounts to the following
claim:
\begin{proposition}
 $O$ contains a (unique) least element.
\end{proposition}
\begin{proof}
 Define the pair of families of sets $(\mathcal{L}_i,\mathcal{L}_j)$ where $\mathcal{L}_i = \bigcap \{\alg_i
\mid (\alg_i,\alg_j) \in O\}$ and $\mathcal{L}_j = \bigcap \{\alg_j \mid
(\alg_i,\alg_j) \in O\}$. It is relatively easy to see that $(\mathcal{L}_i,\mathcal{L}_j)
\leq (\alg_i,\alg_j)$ for all $(\alg_i,\alg_j) \in O$. We argue that
$(\mathcal{L}_i,\mathcal{L}_j)$ is again a pair of algebras and it is in $O$. Thus,
$(\mathcal{L}_i,\mathcal{L}_j)$ is the least element of $O$.
We leave it to the reader to check that $(\mathcal{L}_i,\mathcal{L}_j)$ is again a pair of
algebras. This essentially reduced to the well-known fact that the arbitrary
intersection of algebras is again an algebra.

To show that $(\mathcal{L}_i,\mathcal{L}_j) \in O$ we need to show that it satisfies the
two conditions of the list above. We only check the first one that
$B^{ob}_{\X,i}(K,F) \in \mathcal{L}_i$ for all finite $K,L \subseteq
\acts(A_j \times (T_j,\mathcal{L}_j))$. To this aim fix any finite $K,L
\subseteq \acts(A_j \times (T_j,\mathcal{L}_j))$. To show that
$B^{ob}_{\X,i}(K,F) \in \mathcal{L}_i$ we need to show that
$B^{ob}_{\X,i}(K,F) \in \alg_i$ for all
$(\alg_i,\alg_j) \in O$. This would follow because $(\alg_i,\alg_j)$
also satisfies the first item in the list above, provided that we can
show that $K,L \subseteq \acts(A_j \times (T_j,\alg_j))$. To obtain the
latter we argue that $\acts(A_j \times (T_j,\mathcal{L}_j)) \subseteq
\acts(A_j \times (T_j,\alg_j))$. Because $\mathcal{L}_j \subseteq
\alg_j$ it follows that every measurable set in $A_j \times
(T_j,\mathcal{L}_j)$ is also measurable in $A_j \times (T_j,\alg_j)$.
Hence, a function that is measurable for the domain $A_j \times
(T_j,\mathcal{L}_j)$ is also measurable for the domain $A_j \times
(T_j,\alg_j)$.
\end{proof}

\subsubsection{Proof of Proposition~\ref{p:chara nonred}}

To prove Proposition~\ref{p:chara nonred} we need to perform another
inductive argument over the definition of the terminal sequence from
which the universal choice structure $\U$ was defined. Let us first
recall relevant facts from \iftoggle{preprint}{Sections \ref{preliminary observations} and
\ref{why terminal}}{\ref{preliminary observations} and
\ref{why terminal}} that will be used in the following arguments. For
every $n$ we had an approximant $\Omega_n$ in $\Unc^2$ such that
$\Omega_0$ is the terminal object of $\Unc^2$ and $\Omega_{n + 1} =
\choicePlayers \Omega_n$. Concretely, this meant that $\Omega_{i,0} =
\top$ is the one element uncertainty space and $\Omega_{i,n + 1} =
\choice (A_j \times \Omega_{j,n})$, and similarly with $i$ and $j$
swapped. Given any choice structure $\X = (T,\theta)$, where $T =
(T_i,T_j)$ and $\theta = (\theta_i,\theta_j)$ are in $\Unc^2$, we then
defined the maps $\upsilon_n : T \to \Omega_n$ such that $\upsilon_0 =
\bang_T$ is the unique morphism into the terminal object $\Omega_0$ of
$\Unc^2$ and $\upsilon_{n + 1} = \choicePlayers \upsilon_n \circ
\theta$. If we spell out the latter in terms of measurable functions we
have that
\begin{equation} \label{e:concrete definition upsilon}
 \upsilon_{i,n+1} = \choice(\id_{A_j} \times \upsilon_{j,n}) \circ \theta_i
    \quad \mbox{and} \quad 
 \upsilon_{j,n+1} = \choice(\id_{A_i} \times
\upsilon_{i,n}) \circ \theta_j.
\end{equation}
We then used the properties of the limit $\Omega$ to obtain a the morphism
$\upsilon : T \to \Omega$, which is unique for the property that
$\upsilon_n = \zeta_n \circ \upsilon$. Here, $\zeta_n : \Omega \to
\Omega_n$ is the projection from the limit into the cochain. For
measurable functions this property means that
\begin{equation} \label{e:concrete property upsilon}
 \upsilon_{i,n} = \zeta_{i,n} \circ \upsilon_i
\quad \mbox{and} \quad
 \upsilon_{j,n} = \zeta_{j,n} \circ \upsilon_j.
\end{equation}

\begin{lemma} \label{pain lemma}
 Let $\X = (T_i,T_j,\theta_i,\theta_j)$ be any choice structure and for
every $n \in \omega$ consider $\upsilon_{i,n} : T_i \to \Omega_{i,n}$
as defined above. Then it holds for every set $A \subseteq \Omega_{i,n}$
that $\upsilon_{i,n}^{-1}[A] \in \alg^{ob}_{\X,i}$. The same holds with
$j$ in place of $i$.
\end{lemma}
\begin{proof}
 The proof is an induction on $n$. In the base case we have that
$\Omega_{i,0}$ is the one element uncertainty space. Thus,
$\upsilon_{i,0}^{-1}[A]$ is either the empty set or the whole set of all
types in $T_i$. Thus, it is in the algebra $\alg^{ob}_{\X,i}$.

To prove the inductive step we show below that $\upsilon_{i,n + 1}^{-1}[B^K_L]
\in \alg^{ob}_{\X,i}$ holds for for all $K,L \subseteq \acts (A_j \times
\Omega_{j,n})$, where $B^K_L = \{C \in \Omega_{j,n + 1} \mid C(K)
\subseteq L\}$.
Let us see first how from this it then follows that $\upsilon_{i,n + 1}^{-1}[A]
\in \alg^{ob}_{\X,i}$ holds for every arbitrary set $A \subseteq \Omega_{n + 1,i}$:
Because $\Omega_{i,n + 1}$ is finite and any two
distinct choice functions in $\Omega_{i,n + 1} = \choice (A_j \times
\Omega_{j,n})$ can be separated by some such set $B^K_L$, it follows from
a standard argument that every set $A \subseteq \Omega_{n + 1,i}$ can be
written as a finite union of intersections of sets of the form $B^K_L$
and complements of such sets. Because these unions, intersections and
complements are preserved under taking preimages and $\alg^{ob}_{\X,i}$
is an algebra it follows that $\upsilon_{i,n}^{-1}[A] \in
\alg^{ob}_{\X,i}$.

We now prove that $\upsilon_{i,n + 1}^{-1}[B^K_L]
\in \alg^{ob}_{\X,i}$ holds for for all $K,L \subseteq \acts (A_j \times
\Omega_{j,n})$. First, observe that by the induction hypothesis we have that if $f \in
\acts(A_j \times \Omega_{j,n})$ then $f \circ (\id_{A_j} \times
\upsilon_{j,n}) \in \acts(A_j \times (\Omega_j,\alg^{ob}_{\X,j}))$. This
holds because by the inductive hypothesis the preimage of any set at
$\Omega_{j,n}$ under $\upsilon_{j,n}$ is measurable in the uncertainty
space $(T_j,\alg^{ob}_{\X,j})$. Thus, the function $f \circ
(\id_{A_j} \times \upsilon_{j,n})$ is measurable even when we take its
domain to be $A_j \times (T_j,\alg^{ob}_{\X,j})$.
Then, fix the sets $K,L \subseteq \acts (A_j \times \Omega_{j,n})$ and
define the finite sets
\[
 K' = \{f \circ (\id_{A_j} \times \upsilon_{j,n}) \mid f \in K\}
 \quad \mbox{and} \quad
 L' = \{f \circ (\id_{A_j} \times \upsilon_{j,n}) \mid f \in L\}.
\]
From the observation in the beginning of the paragraph it follows that
that $K', L' \subseteq \acts(A_j \times (T_j,\alg^{ob}_{\X,j}))$.

To show that $\upsilon_{i,n+1}^{-1}[B^K_L] \in \alg^{ob}_{\X,i}$, it suffices by
the definition of $\alg^{ob}_{\X,i}$ to show that
$\upsilon_{i,n+1}^{-1}[B^K_L] = B^{ob}_{\X,i}(K',L')$. Unfolding the
definitions shows that this amounts to the claim that for all $t \in
T_i$ we have $\upsilon_{i,n + 1}(t)(K) \subseteq L$ iff $\theta_i(t)(K')
\subseteq L'$. To see that this is the case consider the following chain
of equivalences:
\begin{align*}
 \upsilon_{i,n + 1}(t)(K) \subseteq L & \mbox{ iff } (\choice(\id_{A_j} \times
\upsilon_{j,n}) \circ \theta_i (t))(K) \subseteq L & \mbox{by
\eqref{e:concrete definition upsilon}} \\
 & \mbox{ iff } \{f \in K \mid f \circ (\id_{A_j} \times \upsilon_{j,n}) \in
\theta_i(t)(K')\}
\subseteq L & \mbox{definition of } \Gamma \\
& \mbox{ iff } \theta_i(t)(K') \subseteq L'
\end{align*}
The last equivalence in this chain can be checked easily using the
definitions of $K'$ and $L'$.
\end{proof}

\begin{proposition} \label{p:real work}
 Let $\X = (T_i,T_j,\theta_i,\theta_j)$ be any choice structure and
$\upsilon = (\upsilon_i,\upsilon_j)$ the unique morphism into the
terminal choice structure $\U = (\Omega_i,\Omega_j,\mu_i,\mu_j)$. Let
$\alg_i$ be the algebra of $\Omega_i$. Then $\alg^{ob}_{\X,i} =
\{\upsilon_i^{-1}[E] \mid E \in \alg_i\}$. A similar claim holds with
$j$ for $i$.
\end{proposition}
\begin{proof}
 Let $\algI_i = \{\upsilon_i^{-1}[E] \mid E \in \alg_i\}$ and $\algI_j =
\{\upsilon_j^{-1}[E] \mid E \in \alg_j\}$. Because taking preimages
preserves intersections and complements it follows that $\algI_i$ and
$\algI_j$ are algebras. Before we argue that $\alg^{ob}_{\X,i} \subseteq
\algI_i$, define
\[
 \algJ_j = \{(\id_{A_j} \times \upsilon_j)^{-1}[E] \mid E \mbox{
measurable in } A_j \times \Omega_j\}.
\]
Observe that $A_j \times (T_j,\algI_j) =
(A_j \times T_j,\algJ_j)$. One can prove both inclusions of this equality by considering the generating cylinders in the product spaces $A_j \times (T_j,\algI_j)$ and $A_j \times \Omega_j$.
Then consider any act $f \in \acts(A_j \times (T_j,\algI_j)) = \acts
(A_j \times T_j,\algJ_j)$. We use Lemma~\ref{factorisation} from \iftoggle{preprint}{Appendix~\ref{choice preserves limits}}{\ref{choice preserves limits}} to show that there is some act $f'
\in \acts (A_j \times \Omega_j)$ such that
\begin{equation} \label{e:factorisation}
 f = f' \circ (\id_{A_j} \times \upsilon_j).
\end{equation}
To see that all assumptions of Lemma~\ref{factorisation} are satisfied, observe first
that by the definition of $\algJ_j$ the function $(\id_{A_j} \times
\upsilon_j)$ is measurable from $(A_j \times T_j,\algJ_j)$ to $A_j \times
\Omega_j$. Moreover, because $f$ is a measurable function from $(A_j
\times T_j,\algJ_j)$ to the discrete uncertainty space $Z$ we have that
$f^{-1}[\{z\}]$ is measurable in $(A_j \times T_j,\algJ_j)$ for every $z
\in Z$. By the definition of $\algJ_j$ this means that for every $z \in
Z$ there is some measurable set $E_z$ in $A_j \times \Omega_j$ such
that $f^{-1}[\{z\}] = (\id_{A_j} \times \upsilon_j)^{-1}[E_z]$.

To prove that $\alg^{ob}_{\X,i} \subseteq \algI_i$ we show that for all
finite $K,L \subseteq \acts(A_j \times (T_j,\algI_j))$ the set
$B^{ob}_{\X,i}(K,L)$ is in $\algI_i$. This, and together with the
analogous statement in which $i$ and $j$ are swapped, shows that 
$\algI_i$ and $\algI_j$ satisfy the properties from
Definition~\ref{observables}. The claim that $\alg^{ob}_{\X,i} \subseteq
\algI_i$ then follows because $\alg^{ob}_{\X,i}$ and $\alg^{ob}_{\X,j}$
are defined to be the smallest algebras which satisfy these properties.

Consider any finite $K,L \subseteq \acts(A_j \times (T_j,\algI_j))$.
Define the finite sets $K',L' \subseteq \acts(A_j \times \Omega_j)$ such
that $K' = \{f' \mid f \in K\}$ and $\{f' \mid f \in L\}$, where $f'$ is
defined from $f$ as in \eqref{e:factorisation}. It follows immediately
from \eqref{e:factorisation} that
\begin{equation} \label{e:prime is good}
 K = \{g \circ (\id_{A_j} \times \upsilon_j) \mid g \in K'\}
\quad \mbox{and} \quad
 L = \{g \circ (\id_{A_j} \times \upsilon_j) \mid g \in L'\}.
\end{equation}

We are then going to show that
\begin{equation} \label{e:morphium}
 B^{ob}_{\X,i}(K,L) = (\choice (\id_{A_j}
\times \upsilon_j) \circ \theta_i)^{-1}[B^{K'}_{L'}],
\end{equation}
where $B^{K'}_{L'} = \{C \in \choice (A_j \times \Omega_j) \mid C(K')
\subseteq L'\}$ is one of the basic measurable sets that generates the
algebra on $\choice(A_j \times \Omega_j)$. Because $\upsilon$ is a
morphism of choice structures it follows that $B^{ob}_{\X,i}(K,L) =
(\mu_i \circ \upsilon_i)^{-1}[B^{K'}_{L'}] =
\upsilon_i^{-1}[\mu_i^{-1}[B^{K'}_{L'}]]$. From this we obtain that
$B^{ob}_{\X,i}(K,L) \in \algI_i = \{\upsilon_i^{-1}[E] \mid E \in
\alg_i\}$ because $\mu_i : \Omega_i \to \choice (A_j \times \Omega_j)$
is measurable and hence $\mu_i^{-1}[B^{K'}_{L'}] \in \alg_i$.

To show \eqref{e:morphium} we unfold the definitions and see that it amounts
to the claim that $\theta_i(t)(K) \subseteq L$ iff $(\choice (\id_{A_j}
\times \upsilon_j) (\theta_i(t)))(K') \subseteq L'$ holds for for all $t
\in T_i$. Using the definition of $\choice$ the right side of the
equivalence in this claim unfolds to the statement
\[
 \{g \in K' \mid g \circ \choice (\id_{A_j} \times \upsilon_j) \in
\theta_i(t)(\{g \circ \choice (\id_{A_j} \times \upsilon_j) \mid g \in
K'\})\} \subseteq L'.
\]
Using the definitions of $K'$ and $L'$ and the equalities from
\eqref{e:prime is good} one can easily see that this is equivalent to
$\theta_i(t)(K) \subseteq L$.

\medskip

To prove the other inclusion fix a measurable set $E \in \alg_i$. We
need to show that $\upsilon_i^{-1}[E] \in \alg^{ob}_{\X,i}$. By the
definition of the algebra $\alg_i$ on the limit $\Omega_i$ it
follows that there is some $n$ such that $E = \zeta_{i,n}^{-1}[A]$ for some
$A \subseteq \Omega_{i,n}$. By Lemma~\ref{pain lemma} we know that
$\upsilon_{i,n}^{-1}[A] \in \alg^{ob}_{\X,i}$. From \eqref{e:concrete
property upsilon} we have that $\upsilon_{i,n} = \zeta_{i,n} \circ
\upsilon_i$ and thus
\[
 \upsilon_{i,n}^{-1}[A] = (\zeta_{i,n} \circ
\upsilon_i)^{-1}[A] = \upsilon_i^{-1}[\zeta_{i,n}^{-1}[A]]
= \upsilon_i^{-1}[E].
\]
With this we can conclude that $\upsilon_i^{-1}[E] \in
\alg^{ob}_{\X,i}$.
\end{proof}

Because the identity morphism is the unique morphism from the universal
choice structure to itself it follows as a simple corollary from
Proposition~\ref{p:real work} that $\alg^{ob}_{\U,i}$ is equal to the
algebra $\alg_i$ of $\Omega_i$. This property has already been observed
in other settings and is called ``minimality'' in \cite{DiTillio08}. 

Before we provide the proof of Proposition~\ref{p:chara nonred} we would
like to mention that the proposition shows that non-redundancy is the
same as the notion of observability from the theory of coalgebras. In
\cite{Jacobs16} a coalgebra is called ``observable'' if the unique
morphism from the coalgebra to the terminal coalgebra is monic. In
\cite{Rutten00} such coalgebras are called ``simple''. 



\begin{proof}[Proof of Proposition~\ref{p:chara nonred}]
For the direction from left to right assume that $\X$ is non-redundant
and consider the unique morphism $\upsilon = (\upsilon_i,\upsilon_j)$
from $\X = (T_i,T_j,\theta_i,\theta_j)$ into the universal choice
structure $\U$. To show that $\upsilon_i$ is injective consider $t,t'
\in T_i$ such that $t \neq t'$. Because $\X$ is non-redundant
$\alg^{ob}_{\X,i}$ is separating, and thus there is some $E \in
\alg^{ob}_{\X,i}$ with $t \in E$ and $t' \notin E$. By
Proposition~\ref{p:real work} there is then some $E'$ that is measurable
in $\Omega_i$ such that $E = \upsilon_i^{-1}[E']$. But then it must be
the case that $\upsilon_i(t) \neq \upsilon_i(t')$, because otherwise
$\upsilon_i(t') = \upsilon_i(t) \in E'$ and thus $t' \in
\upsilon_i^{-1}[E'] = E$. The same reasoning works for $j$ in place of
$i$.

For the other direction assume that the $\upsilon_i : T_i \to \Omega_i$
is injective. To show that $\alg^{ob}_{\X,i}$ is separating on $T_i$
consider any $t,t' \in T_i$ with $t \neq t'$. Because $\upsilon_i$ is
injective we have $\upsilon_i(t) \neq \upsilon_i(t')$.
The algebra on $\Omega_i$, which is defined as the limit algebra
of the cochain $(\Omega_{i,n},\xi_{i,n})_{n \in \omega}$ for discrete spaces
$\Omega_{i,n}$, is separating because any two distinct infinite sequences have distinct projections to some finite level $n$, where they can be separated.
Thus there must be a measurable set $E$ in $\Omega_i$ such that $t \in
E$ and $t' \notin E$. By Proposition~\ref{p:real work} the set
$\upsilon_i^{-1}[E]$ is in $\alg^{ob}_{\X,i}$ and obviously this set
separates $t$ and $t'$. The same argument can be used with $j$ in place
of $i$.
\end{proof}

\section{Proofs for Section~\ref{preference structures}}
\label{proofs for prel}

\subsection{Preliminary observations}

It is easy to check that $\justprel$, and hence also $\prel$ and
$\prelPlayers$ are functors, that is, they preserve identities and
composition. The construction of the universal preference structure can
then be carried out analogously to the construction for the universal
choice structure given above. In fact it is only required to reprove a
variant of Theorem~\ref{justchoice preserves limits} for the functor
$\justprel$, which is done implicitly in the proofs of Section~3 from
\cite{DiTillio08}. 

\subsection{Proposition~\ref{justmaxi injective}}

We now argue that the map $\justmaxi_X : \justprel X \to \justchoice X$
is injective at every set $X$. To show this, assume we have two
preference relations $\pref$ and $\pref'$ over $X$ such that
$\justmaxi_X({\pref}) = \justmaxi_X({\pref'})$. We need to argue that
$x \pref x'$ iff $x \pref' x'$ for all $x,x' \in X$. Since the situation
is symmetric it suffices to check one direction. Hence assume that $x
\pref x'$. We want to show that $x \pref' x'$. It suffices to consider
the case where $x \neq x'$ because otherwise $x \pref' x'$ follows
because $\pref'$ is reflexive.

As $x \neq x'$ and $x \pref x'$ it follows that it can not be the case
that $x' \pref x$, because otherwise there would be a contradiction with
the anti-symmetry of $\pref$. As we explain in Remark~\ref{why posets}
this use of anti-symmetry is crucial. As $x \pref x'$ and not $x' \pref
x$ it follows that $x'$ is the only maximal element of the set
$\{x,x'\}$ in the relation $\pref$. Hence $\justmaxi_X({\pref})(\{x,x'\}) =
\{x'\}$.

By the assumption that $\justmaxi_X({\pref}) = \justmaxi_X({\pref'})$ it
follows that $\justmaxi_X({\pref'})(\{x,x'\}) = \{x'\}$. But this is
only possible if $x \pref' x'$, which is what we had to show.

\subsection{Proposition~\ref{maxi natural}}
\label{proof maxi natural}

Proposition~\ref{maxi natural} states that $\maxi$ is a natural
transformation from the functor $\prel$ to the functor $\choice$. It is
easy to check that this reduces to the claim that $\justmaxi$ is a
natural transformation from $\justprel$ to $\justchoice$. This means
that we need to show that for every function $f : X \to Y$ it holds that
$\justmaxi_X \circ \justprel f = \justchoice f \circ \justmaxi_Y$. Note
that here the $Y$ and $X$ are swapped because $\justprel$ and
$\justchoice$ are contravariant functors.

Fix any function $f : X \to Y$, preference relation ${\pref} \in
\justprel Y$ and finite set $K \subseteq X$. We have
to show that
\[
 \justmaxi_X(\justprel f(\pref))(K) = \justchoice f (\justmaxi_Y(\pref))(K) .
\]

For the left-to-right inclusion take any $x \in \justmaxi_X(\justprel
f(\pref))(K)$. We need to show that $x \in \justchoice f
(\justmaxi_Y(\pref))(K)$. This means we want to show that $x \in
f^{-1}[\justmaxi_Y(\pref)(f[K])] \cap K$. Our assumption that $x \in
\justmaxi_X(\justprel f(\pref))(K)$ means that $x$ is a maximal element
of the set $K$ in the order ${\pref^f} = \justprel f(\pref)$. Hence $x
\in K$ and it remains to show that $x \in
f^{-1}[\justmaxi_Y(\pref)(f[K])]$, which means that $f(x)$ is a maximal
element of the set $f[K]$ in the order $\pref$. So consider any other
element of $f[K]$, which must be of the form $f(x')$ for some $x' \in
K$, and assume that $f(x) \pref f(x')$. We need to show that then also
$f(x') \pref f(x)$. From equation \eqref{defining justprel} in
Section~\ref{di tillio} defining $\pref^f$ it follows that $x \pref^f
x'$. Then we can use that $x$ is a maximal element in $\pref^f$ to
conclude that $x' \pref^f x$. Using \eqref{defining justprel} again, we
then obtain the required $f(x') \pref f(x)$.

Now consider the right-to-left inclusion. Take any $x \in
f^{-1}[\justmaxi_Y(\pref)(f[K])] \cap K$. We need to show that $x \in
\justmaxi_X(\justprel f(\pref))(K)$, which means that $x$ is a maximal
element of the set $K$ in the order ${\pref^f} = \justprel f(\pref)$.
Clearly $x \in K$. To show that $x$ is a $\pref^f$-maximal element in
$K$ pick any other $x'$ in $K$ with $x \pref x'$. We need to show that
$x' \pref x$. From $x \pref x'$ it follows with \eqref{defining
justprel} that $f(x) \pref f(x')$. We now use that $x \in
f^{-1}[\justmaxi_Y(\pref)(f[K])]$. From this it follows that $f(x) \in
\justmaxi_Y(\pref)(f[K])$. This means that $f(x)$ is maximal in the set
$f[K]$. Because $x' \in K$ it holds that also $f(x')$ is in the set
$f[K]$. By the maximality of $f(x)$ in $f[K]$ it follows from $f(x)
\pref f(x')$ that $f(x') \pref f(x)$. With \eqref{defining justprel} we
obtain $x' \pref x$.

\subsection{Theorem~\ref{upsilon injective}}

In the proof of Theorem~\ref{upsilon injective} we need a further
concept from category theory which is a generalization of monic
morphism. A family of morphism $(f_j : Y \to X_j)_{j \in J}$ for any
index set $J$ is \emph{jointly monic} if for all further morphisms $g,h
: T \to Y$ it holds that if $f_j \circ g = f_j \circ h$ for all $j \in
J$ then already $g = h$.

Families of monic morphism are closely related to limits of cochains.
Using the universal property of the limit $X_\omega$ with projections
$\zeta_n$ of a cochain $(X_n,f_n)_{n \in \omega}$, it is easy to show
that the family of all projections $(p_n)_{n \in \omega}$ is jointly
monic. Moreover, if we have another object $T$ with a jointly monic
family $(g_n : T \to X_n)_{n \in \omega}$ such that $g_n = f_n \circ
g_{n + 1}$ for all $n$ then the unique morphism $u : T \to X_\omega$
that exists because of the universal property of the limit $X_\omega$ is
monic.

To prove Theorem~\ref{upsilon injective}, let $\U' = (\Omega',\mu')$ be
the universal preference structure, presented as coalgebra for
$\prelPlayers$. We assume that $\Omega'$, $\mu'$, $\Omega'_n$,
$\zeta'_n$, \dots are defined in the same way as the objects $\Omega$,
$\mu$, $\Omega_n$, $\zeta_n$, \dots are defined in \iftoggle{preprint}{Section~\ref{proofs
for choices}}{\ref{proofs for choices}} for the universal choice structure, just using
$\prelPlayers$ instead of $\choicePlayers$.

Then consider the choice structure $\embedded{\U'} =
(\Omega',\maxiPlayers_{\Omega'} \circ \mu')$ where $\maxiPlayers$ is a
natural transformation from $\prelPlayers$ to $\choicePlayers$ that is
defined such that it applies $\maxi$ componentwise. Note that this
definition of $\embedded{\U'}$ corresponds to the one given in
Section~\ref{embedding}. Let $\upsilon$ be the unique morphism from
$\embedded{\U'}$ to the universal choice structure $\U$ that exists
according to Theorem~\ref{omega terminal}.

The claim of Theorem~\ref{upsilon injective} is that this $\upsilon$ is
monic. Because in Theorem~\ref{omega terminal} $\upsilon$ was obtain
using the universal property of the limit $\Omega$ from the family of
approximations $(\upsilon_n : \Omega'_n \to \Omega_n)_{n \in \omega}$ it
suffices to show that this family is jointly monic.

Define morphisms $\delta_0 = \bang_{\Omega'_0} : \Omega'_0 \to \Omega_0$
and inductively $\delta_{n + 1} = \choicePlayers \delta_n \circ
\maxiPlayers_{\Omega'_n} : \Omega'_{n + 1} \to \Omega_{n + 1}$. One can
show with an induction over $n$ that all these $\delta_n$ are monic. The
base case this holds because $\Omega'_0$ is the terminal object of
$\Unc^2$ and in the inductive step we use Proposition~\ref{justmaxi
injective} and the fact hat $\choicePlayers$ preserves monics, which we
show in \iftoggle{preprint}{Section~\ref{choice monic}}{\ref{choice monic}}. The latter needs that the image of
the injective measurable function $\delta_n : \Omega'_n \to \Omega_n$
that is preserved has the discrete algebra. This is the case because one
can show that if $A_i$ and $A_j$ are finite then so are all the
$\Omega'_n$.

We then prove by induction on $n$ that
\begin{equation} \label{last equation}
 \upsilon_n = \delta_n \circ \zeta'_n.
\end{equation}
It follows that the $\upsilon_n$ are jointly monic because the
$\zeta'_n$ are projections out of the limit $\Omega'$ and hence jointly
monic and the $\delta_n$ are all monic.

For the base case, of \eqref{last equation}, we have that $\upsilon_0 =
\delta_0 \circ \zeta'_0$ because both morphism map to the terminal
object $\Omega_0$ of $\Unc^2$.

For the inductive step we use the following computation:
\begin{align*}
 \upsilon_{n + 1} & = \choicePlayers \upsilon_n \circ
\maxiPlayers_{\Omega'} \circ \mu' & \mbox{definition of } \upsilon_{n +
1} \\
 & = \choicePlayers (\delta_n \circ \zeta'_n) \circ
\maxiPlayers_{\Omega'} \circ \mu' & \mbox{induction hypothesis} \\
 & = \choicePlayers \delta_n \circ \choicePlayers \zeta'_n \circ
\maxiPlayers_{\Omega'} \circ \mu' & \choicePlayers \mbox{ functor} \\
 & = \choicePlayers \delta_n \circ \maxiPlayers_{\Omega'_n} \circ
\prelPlayers \zeta'_n \circ \mu' & \maxiPlayers \mbox{ natural
transformation} \\
 & = \choicePlayers \delta_n \circ \maxiPlayers_{\Omega'_n} \circ
\tau'_{n + 1} \circ \mu' & \mbox{definition of } \tau'_{n + 1} \\
 & = \choicePlayers \delta_n \circ \maxiPlayers_{\Omega'_n} \circ
\zeta'_{n + 1} & \mbox{uniqueness property of } \mu' \\
 & = \delta_{n + 1} \circ \zeta'_{n + 1} & \mbox{definition of }
\delta'_{n + 1}
\end{align*}

\end{document}